\documentclass[11pt, fleqn]{article}
\usepackage[english]{babel}

\usepackage[lmargin=1.1in,rmargin=1.1in,bottom=1.3in,top=1.3in,
twoside=False]{geometry}

\usepackage{mathrsfs}
\usepackage{relsize,xspace}
 \usepackage{xcolor}
 \usepackage{mathtools}
 \usepackage{todonotes}
 \usepackage{comment}
\usepackage{microtype}
\usepackage{amsmath}
\usepackage{amssymb}
\usepackage{amsfonts}
\usepackage{stmaryrd}
\usepackage{bm}
\usepackage{tikz}
\usepackage{refcount}
\usepackage{wrapfig}

\usepackage{marginnote}

\definecolor{sz}{rgb}{0.1,0.2,0.6}
\definecolor{blue}{rgb}{0.1,0.2,0.5}
\definecolor{brown}{rgb}{0.6,0.6,0.2}

\usepackage[amsmath,thmmarks,hyperref]{ntheorem}
\usepackage{cleveref}

\crefformat{page}{#2page~#1#3}%
\Crefformat{page}{#2Page~#1#3}%
\crefformat{equation}{#2(#1)#3}%
\Crefformat{equation}{#2(#1)#3}%
\crefformat{figure}{#2Figure~#1#3}%
\Crefformat{figure}{#2Figure~#1#3}%
\crefformat{section}{#2Section~#1#3}
\Crefformat{section}{#2Section~#1#3}
\crefformat{appendix}{#2Appendix~#1#3}
\Crefformat{appendix}{#2Appendix~#1#3}
\crefformat{chapter}{#2Chapter~#1#3}
\Crefformat{chapter}{#2Chapter~#1#3}
\crefformat{chapter*}{#2Chapter~#1#3}
\Crefformat{chapter*}{#2Chapter~#1#3}
\crefformat{part}{#2Part~#1#3}
\Crefformat{part}{#2Part~#1#3}
\crefformat{enumi}{#2(#1)#3}
\Crefformat{enumi}{#2(#1)#3}

\usepackage{enumerate}

\usepackage{latexsym}

\theoremnumbering{arabic}
\theoremstyle{plain}
\theoremsymbol{}
\theorembodyfont{\itshape}
\theoremheaderfont{\normalfont\bfseries}
\theoremseparator{.}

\newtheorem{theorem}{Theorem}
\crefformat{theorem}{#2Theorem~#1#3}
\Crefformat{theorem}{#2Theorem~#1#3}
\newcommand{\newtheoremwithcrefformat}[2]{%
  \newtheorem{#1}[lemma]{#2}%
  \crefformat{#1}{##2\MakeUppercase#1~##1##3}%
  \Crefformat{#1}{##2\MakeUppercase#1~##1##3}%
}
\newcommand{\newseptheoremwithcrefformat}[2]{%
  \newtheorem{#1}{#2}%
  \crefformat{#1}{##2\MakeUppercase#1~##1##3}%
  \Crefformat{#1}{##2\MakeUppercase#1~##1##3}%
}

\newseptheoremwithcrefformat{lemma}{Lemma}
\newtheoremwithcrefformat{proposition}{Proposition}
\newtheoremwithcrefformat{observation}{Observation}
\newtheoremwithcrefformat{conjecture}{Conjecture}
\newtheoremwithcrefformat{corollary}{Corollary}
\newseptheoremwithcrefformat{claim}{Claim}
\newseptheoremwithcrefformat{goal}{Goal}
\theorembodyfont{\upshape}
\newtheoremwithcrefformat{example}{Example}
\newtheoremwithcrefformat{remark}{Remark}
\newseptheoremwithcrefformat{definition}{Definition}

\theoremstyle{nonumberplain}
\theoremheaderfont{\scshape}
\theorembodyfont{\normalfont}
\theoremsymbol{\ensuremath{\square}}

\newtheorem{proof}{Proof}

\theoremsymbol{\ensuremath{\lrcorner}}

\renewcommand{\subset}{\subseteq}
\renewcommand{\setminus}{-}

\usepackage{booktabs}   
\usepackage{subcaption} 
                        
\usepackage[absolute]{textpos}

\usepackage{relsize,xspace}
 \usepackage{xcolor}
 \usepackage{mathtools}
 \usepackage{todonotes}
 \usepackage{comment}
\usepackage{microtype}
\usepackage{amsmath}
\usepackage{amssymb}
\usepackage{amsfonts}
\usepackage{mathrsfs}
\usepackage{bm}
\usepackage{tikz}
\usetikzlibrary{shapes,calc,fit}

\usepackage[normalem]{ulem}
\newcommand\stout{\bgroup\markoverwith{\textcolor{red}{\rule[0.5ex]{2pt}{1.4pt}}}\ULon}
\newcommand\pstout{\bgroup\markoverwith{\textcolor{violet}{\rule[0.5ex]{2pt}{1.4pt}}}\ULon}

\usepackage{refcount}
\usepackage{wrapfig}

\usepackage{marginnote}

\tikzstyle{vertex}=[circle,inner sep=0.5,minimum size%
=5mm,semithick,fill=black!20, draw=black]
\tikzstyle{S}=[dash pattern=on 2pt off 1pt,blue,xshift=2mm]
\tikzstyle{G}=[black,bend left=10]
\tikzstyle{U}=[red,bend right=10,thick,densely dotted]
\tikzstyle{T}=[green!70!black,densely dashed,very thick]

\newcommand{\col}{\mathrm{col}}
\newcommand{\adm}{\mathrm{adm}}
\newcommand{\tw}{\mathrm{tw}}
\newcommand{\SReach}{\mathrm{SReach}}
\newcommand{\CCC}{\mathscr{C}}
\newcommand{\DDD}{\mathscr{D}}
\newcommand{\FPT}{\textsc{FPT}}
\newcommand{\AWs}{\textsc{AW}\ensuremath{[\ast]}} 

\newcommand{\Oh}{\mathcal{O}}

\newcommand{\Pp}{\mathcal{P}}
\newcommand{\N}{\mathbb{N}}
\renewcommand{\phi}{\varphi}
\renewcommand{\epsilon}{\varepsilon}
\newcommand{\str}{\mathbb}
\newcommand{\strA}{\str{A}}
\newcommand{\FO}{\ensuremath{\mathrm{FO}}\xspace}
\newcommand{\MSO}{\ensuremath{\mathrm{MSO}}\xspace}
\newcommand{\CMSO}{\ensuremath{\mathrm{CMSO}}\xspace}
\newcommand{\FOs}{\FO[+1]}
\newcommand{\FOsi}{\FO[+1\mathit{-inv}]}
\newcommand{\FOo}{\FO[{<}]}
\newcommand{\FOoi}{\FO[{<}\mathit{-inv}]}
\newcommand{\MSOsi}{\MSO[+1\mathit{-inv}]}

\newcommand{\MSOoi}{\MSO[{<}\mathit{-inv}]}
\newcommand{\MSOoneoi}{\MSO_1[{<}\mathit{-inv}]}
\newcommand{\MSOtwooi}{\MSO_2[{<}\mathit{-inv}]}
\newcommand{\minor}{\preccurlyeq}
\renewcommand{\mid}{\,:\,}
\newcommand{\norm}[1]{\lVert#1\rVert}

\newcommand{\npprob}[5]{%
  \begin{center}\normalfont\fbox{%
      \begin{tabular}[t]{rp{#1}}%
        \multicolumn{2}{l}{#2}\\%
        \textit{Input:} & #3\\%
        \textit{Parameter:} & #4\\%
        \textit{Problem:} & #5%
      \end{tabular}}%
\end{center}}

\renewcommand{\leq}{\leqslant}
\renewcommand{\geq}{\geqslant}
\renewcommand{\le}{\leqslant}

\newcommand{\MC}{\mathrm{MC}}
\newcommand{\LLL}{\mathcal{L}}
\newcommand{\Cdegd}{\DDD_d}

\newcommand{\ptxt}[1]{\textcolor{violet}{#1}}
\newcommand{\ar}{\mathrm{ar}}
\newcommand{\FTP}{\textsc{FTP}}

\hypersetup{ocgcolorlinks=true,colorlinks=true,linkcolor={blue}, citecolor={brown}}

\title{Model-Checking on Ordered Structures\thanks{This paper
  subsumes the results of
  \cite{heuvel17,engelmann2012first,eickmeyer2013model,eickmeyer2016model,KreutzerPRS16}.\newline
  S.~Kreutzer and R.~Rabinovich are supported
  by the European Research Council (ERC) under the European Union's Horizon
  2020 research and innovation programme (ERC Consolidator Grant DISTRUCT,
  grant agreement No.\ 648527). P.~Ossona de Mendez is supported by
    grant ERCCZ LL-1201 and CE-ITI, and by the European Associated
    Laboratory ``Structures in Combinatorics'' (LEA STRUCO) P202/12/G061.
  M.\ Pilipczuk and S.\ Siebertz are supported by the National Science
  Centre of Poland via POLONEZ grant agreement UMO-2015/19/P/ST6/03998,
  which has received funding from the European Union's Horizon 2020
  research and innovation programme (Marie Sk\l odowska-Curie grant
  agreement No.\ 665778). D.A.~Quiroz is supported by CONICYT,
    PIA/Concurso Apoyo a Centros Cient\'\i ficos y Tecnol\'ogicos de
    Excelencia con Financiamiento Basal AFB170001.
}}

\author{Kord Eickmeyer\\Technische Universit\"at Darmstadt, Germany, \\\texttt{eickmeyer@mathematik.tu-darmstadt.de}
\and
Jan van den Heuvel\\London School of Economics and Political Science, United Kingdom\\\texttt{j.van-den-heuvel@lse.ac.uk}
\and
Ken-ichi Kawarabayashi\\National Institute of Informatics, 
Yapan\\\texttt{k\_keniti@nii.ac.jp}
\and
Stephan Kreutzer\\Technische Universit\"at Berlin, Germany\\
\texttt{stephan.kreuzter@tu-berlin.de}
\and
Patrice Ossona~de~Mendez\\Centre d'Analyse et de Math\'ematiques Sociales (CNRS, UMR
    8557), Paris, France\\ and \\Charles University Prague, Czech Republic\\\texttt{pom@ehess.fr}    
\and
Micha\l~Pilipczuk\\University of Warsaw, Poland\\
\texttt{michal.pilipczuk@mimuw.edu.pl}
\and
Daniel~A. Quiroz\\Universidad de Chile, Chile\\
\texttt{dquiroz@cmm.uchile.cl}
\and
Roman Rabinovich\\Technische Universit\"at Berlin, Germany
\\\texttt{roman.rabinovich@tu-berlin.de}
\and 
Sebastian Siebertz\\Humboldt-Universit{\"a}t zu Berlin, Germany\\ \texttt{siebertz@informatik.hu-berlin.de}}

\begin{document}
\maketitle

  \begin{textblock}{5}(12.75,14.15)
    \includegraphics[width=32px]{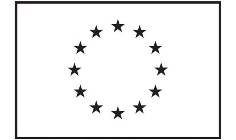}%
  \end{textblock}

\newpage
\begin{abstract}
  We study the model-checking problem for first- and monadic second-order
  logic on finite relational structures. The problem of verifying whether a
  formula of these logics is true on a given structure is considered
  intractable in general, but it does become tractable on interesting
  classes of structures, such as on classes whose Gaifman graphs have
  bounded treewidth.
  In this paper we continue this line of research and study model-checking
  for first- and monadic second-order logic in the presence of an ordering
  on the input structure. We do so in two settings: the general ordered
  case, where the input structures are equipped with a fixed order or
  successor relation, and the order invariant case, where the formulas may
  resort to an ordering, but their truth must be independent of the
  particular choice of order.
  In the first setting we show very strong intractability results for most
  interesting classes of structures. In contrast, in the order invariant
  case we obtain tractability results for order-invariant monadic
  second-order formulas on the same classes of graphs as in the unordered
  case. For first-order logic, we obtain tractability of
  successor-invariant formulas on classes whose Gaifman graphs have bounded
  expansion. Furthermore, we show that model-checking for order-invariant
  first-order formulas is tractable on coloured posets of bounded width.
\end{abstract}

\section{Introduction}
\label{sec:intro}

Pinpointing the exact complexity of the model-checking problem for
first-order and monadic second-order logic has been the object of a large
body of research. The model-checking problem for a logic~$\LLL$, denoted
$\MC(\LLL)$, is the problem of deciding for a given finite
structure~$\strA$ and a formula $\phi\in\LLL$ whether $\strA$ is a model of
$\phi$; in symbols $\strA\models\phi$. We will denote $\MC(\LLL)$
restricted to a class~$\CCC$ of input structures as $\MC(\LLL,\CCC)$.

Vardi~\cite{Vardi82} proposed to distinguish the complexity of the
model-checking problem into \emph{data}, \emph{formula}, and
\emph{combined} complexity, depending on whether we treat the structure
$\strA$ (the data) as input while considering $\phi$ as fixed, the formula
as input while considering $\strA$ as fixed, or considering both $\strA$
and~$\phi$ as part of the input. As shown by Vardi, for any fixed formula
$\phi\in\FO$ of size $\vert\phi\vert$ the model-checking problem is
solvable in polynomial time $n^{\mathcal{O}(\vert\phi\vert)}$, i.e.\@ the
data complexity of $\MC(\FO)$ is in \textsc{Ptime}. On the other hand, the
formula complexity and combined complexity of first-order logic is
\textsc{Pspace}-complete already on a fixed \mbox{2-element}
structure~\cite{ChandraM77}. Evaluating a fixed formula of monadic
second-order logic belongs to the polynomial time hierarchy. (And for each
level $\Sigma_i^p$ and $\Pi_i^p$ there exists an \MSO-formula whose
model-checking problem is complete for that
level~\cite{stockmeyer1976polynomial}.) Again, the formula complexity and
combined complexity of monadic second-order logic is
\textsc{PSpace}-complete.

A more fine-grained analysis of model-checking complexity can be achieved
through the lens of parameterised complexity. In this framework, the
model-checking problem $\MC(\LLL)$ for a logic~$\LLL$ is said to be
\emph{fixed-parameter tractable} if it can be solved in
time $f(|\phi|)\cdot|\strA|^c$, for some function~$f$ (usually required to
be computable) and a constant $c$ independent of~$\phi$ and $\strA$. The
complexity class FPT of all fixed-parameter tractable problems is the
parameterised analogue to \textsc{Ptime} as a model of efficient
solvability. Hence, parameterised complexity lies somewhere between data
and combined complexity, in that the formula is not taken to be fixed and
yet has a different influence on the complexity than the structure. Already
the model-checking problem for first-order logic is complete for the
parameterised complexity class $\AWs$, which is conjectured and widely
believed to strictly contain the class FPT. Thus it is widely believed that
model-checking for first-order logic (and thus also for monadic
second-order logic) is not fixed-parameter tractable.


Perhaps the most famous result on the parameterised complexity of
model-checking is Courcelle's theorem \cite{Courcelle90}, which states that
every algorithmic property on graphs definable in monadic second-order
logic (with quantification over edge sets) can be evaluated in linear time
on any class of graphs of bounded treewidth. An equivalent statement is
that $\MC(\MSO,\CCC)$ is fixed-parameter tractable via a linear-time algorithm for any class
$\CCC$ of bounded treewidth. This result was followed by a similar result
for monadic second-order logic with only quantification over vertex sets on
graph classes of bounded clique-width~\cite{CourcelleMakRot00}. It was shown in
\cite{KreutzerT10,KreutzerT10b} that Courcelle's theorem cannot be extended
in full generality much beyond bounded treewidth.


For first-order logic, Seese \cite{Seese96} proved that first-order model-checking is fixed-parameter tractable on any class of graphs of bounded
degree. This result was the starting point of a long series of papers
establishing tractability results for first-order model-checking on sparse
classes of graphs; see e.g.\
\cite{FrickG01,FlumFG02,DawarGK07,dvovrak2013testing,grohe2014deciding},
and see \cite{grohe2011methods} for a survey. This line of research
culminated in the theorem of Grohe et al.\ \cite{grohe2014deciding} stating
that for any nowhere dense class $\CCC$ of graphs we have
  $\MC(\FO,\CCC)\in\FPT$. Moreover, for classes of graphs that are closed
under taking subgraphs, this yields a precise characterisation of
tractability for first-order model-checking~\cite{dvovrak2013testing}.


So far, most of the work on algorithmic meta-theorems has focused on
unordered structures. Many of the results mentioned above rely on locality
theorems for first-order logic, such as Gaifman's locality theorem
\cite{Gaifman82}, and the applied techniques do not readily extend to
ordered structures. In this paper we study the complexity of first-order
model-checking on structures where an ordering is available to be used in
formulas. We do so in two different settings. The first is that the input
structures are equipped with a fixed order or with a fixed successor
relation. (A successor relation is a directed Hamiltonian path on
the universe of the structure.) We show that first-order logic on ordered
structures as well as on structures with a successor relation is
essentially intractable on nearly all interesting classes.

The other case we consider is an order-\emph{invariant} or a
successor-\emph{invariant} logic. In order-invariant logics, we are allowed to use
an order relation in the formulas, but whether the formula is true in a
given structure must not depend on the particular choice of order.

It is easily seen that the expressive power of order-invariant \MSO is
greater than that of plain \MSO, as, e.g. with an order we can
formalise in \MSO that a structure has an even number of elements, a
property not definable without an order. In fact, the expressive power of
order-invariant \MSO is even greater than the expressive power of the
extension of \MSO with counting quantifiers \CMSO~\cite{GanzowR08}. Over
restricted classes of structures, order-invariant \MSO and \CMSO have the
same expressive power (see e.g.\ \cite{Courcelle96}). This holds true for
successor-invariant \MSO as well, as an order is definable from a successor
relation via \MSO.

For monadic second-order logic we are able to show that order-invariant
\MSO is tractable on essentially the same classes of graphs as plain \MSO,
i.e.\@ we can increase the expressive power without restricting the
tractable cases. To be precise, we show that the model-checking problem for
order-invariant \MSO on classes of  graphs of bounded clique-width is
fixed-parameter tractable. Furthermore, combining the result of Courcelle
\cite{Courcelle90} and results that one can add the edges of 
a successor relation to a graph of bounded treewidth without 
increasing its treewidth too much, we get
that model-checking for order-invariant \MSO (with quantification over edge
sets) on classes of graphs of bounded treewidth is fixed-parameter tractable.

Also successor- and order-invariant first-order logic have both been
studied intensively in the literature, see e.g.\
\cite{potthoff94,Otto00,BenediktSeg2009,rossman2007successor,Rossman03,elberfeld2016order}.
However, the difference between the expressive powers of order-invariant,
successor-invariant, and plain \FO on various classes of structures remains
largely unexplored. An unpublished result of Gurevich states that the
expressive power of order-invariant \FO is stronger than that of plain \FO
(see e.g.\ Theorem~5.3 of \cite{libkin2013elements} for a presentation of
the result). Rossman \cite{rossman2007successor} proved the stronger result that that
successor-invariant \FO is more expressive than plain \FO.
The construction of \cite{rossman2007successor} creates dense instances
though, and no separation between successor-invariant \FO and plain \FO is
known on sparse classes, say on classes of bounded expansion. On the other
hand, collapse results in this context are known only for very restricted
settings. It is known that order-invariant \FO collapses to plain \FO on
trees~\cite{BenediktSeg2009,Niemistro05} and on graphs of bounded
treedepth~\cite{EickmeyerEH14}. Moreover, order-invariant \FO is a subset
of \MSO on graphs of bounded degree and on graphs of bounded treewidth
\cite{BenediktSeg2009}, and, more generally, on decomposable graphs in the
sense of~\cite{elberfeld2016order}.

We show that, up to a narrow gap, the model-checking results for plain
\FO carry over to successor-invariant \FO. In
particular, we show that model-checking successor-invariant~\FO is
fixed-parameter tractable on any class of graphs of bounded expansion.
Classes of bounded expansion generalise classes with excluded topological
minors, and form a natural meta-class one step below nowhere dense classes
of graphs. More precisely, we show that if $\CCC$ is a class of structures
of bounded expansion, then model-checking for successor-invariant
first-order formulas on $\CCC$ can be solved in time
$f(|\phi|)\cdot n\cdot\alpha(n)$, where $n$ is the size of the universe of
the given structure, $f$ is some function, and $\alpha(\cdot)$ is the
inverse Ackermann function. Note that model-checking for plain \FO can be done in linear time on classes of bounded expansion
\cite{dvovrak2013testing}, thus the running time of our algorithm is very
close to the best known results for plain \FO.

The natural way of proving tractability for successor-invariant \FO on a
specific class $\CCC$ of graphs is to show that any given graph $G\in\CCC$
can be augmented by a new set $S$ of coloured edges which form a successor
relation on $V(G)$ such that $G+S$ falls within a class $\DDD$ of graphs on
which plain \FO is tractable. In this way, model-checking for
successor-invariant \FO on the class~$\CCC$ is reduced to the
model-checking problem for \FO on $\DDD$. The main problem is how to
construct the set of augmentation edges $S$. For classes of bounded
expansion, to construct such an edge set, we rely on a characterisation of
bounded expansion classes by \emph{generalised colouring numbers}. The
definition of these graph parameters is roughly based on measuring
reachability properties in a linear vertex ordering of the input graph. Any
such ordering yields a very weak form of decomposition of a graph in terms
of an elimination tree. The main technical contribution of this paper is
that we find a way to control these elimination trees so that we can use
them to define, in a first step, a set $F$ of new edges with the following
properties: a) $F$ forms a spanning tree of the input graph~$G$, b) $F$ has
maximum degree at most $3$, and c) after adding all the edges of $F$ to the
graph, the increase in the colouring numbers is bounded. In a second step, from the
bounded degree spanning tree we will construct a successor relation $S$ as
desired.

This construction, besides its use in this paper, yields a new insight into
the elimination trees generated by colouring numbers. We believe it may
prove useful in future research as well.

As mentioned before, the tractability of model-checking for \FO on sparse graphs is well understood, while only few results are
available for classes of dense graphs. We review some known results for \FO
model-checking on dense graph classes in Section~\ref{sec:dense} and show
that a result by Gajarsk\'y et al.~\cite{GajarskyHLOORS15} carries over to
order-invariant \FO.

\paragraph*{Organisation of the paper.}
In Section~\ref{sec:prelims} we fix the terminology and notation used
throughout the paper. In Section~\ref{sec:order} we study the case of
ordered structures, i.e.\@ structures equipped with a fixed order or
successor relation. Order-invariant \MSO is considered in
Section~\ref{sec:order-inv}. We recall the notions from the theory of
sparse graphs, in particular the generalised colouring numbers, and prove
tractability of successor-invariant \FO on bounded expansion classes in
Section~\ref{sec:mc-sparse}. Finally, in Section~\ref{sec:dense} we
consider order-invariant \FO on posets of bounded width and other dense
classes of structures.

\paragraph*{Acknowledgements.}
We thank Christoph Dittmann and Viktor Engelmann for many hours of fruitful
discussions.


\section{Preliminaries}\label{sec:prelims}

\paragraph*{General notation.}
By $\N$ we denote the set of nonnegative integers. For a set $X$, by
$\binom{X}{2}$ we denote the set of unordered pairs of elements of $X$,
that is, $2$-element subsets of $X$. By $\alpha(\cdot)$ we denote the
inverse Ackermann function, i.e.\ the inverse of the function
$n\mapsto A(n,n)$ with
\[A(m,n)\coloneqq\begin{cases}
  n+1,&\text{if $m=0$};\\
  A(m-1,1),&\text{if $m>0$ and $n=0$};\\
  A\big(m-1,A(m,n-1)\big),&\text{if $m,n>0$}.
\end{cases}\]

\paragraph*{Relational structures.}
We consider finite structures over finite signatures that contain only
relation symbols and constant symbols. Hence a signature
$\tau=\{R_1,\ldots,R_k,c_1,\ldots,c_s\}$ is a finite set of relation
symbols $R_i$ and constant symbols $c_i$, where each relation symbol
$R\in\tau$ is assigned an \emph{arity $\ar(R)$} (arities are part of the signature). A $\tau$-structure
$\strA=\bigl(V(\strA),R_1(\strA),\ldots,R_k(\strA),
c_1(\strA),\ldots,c_s(\strA)\bigr)$ consists of a set $V(\strA)$, the
\emph{universe} of $\strA$, for each $R_i\in\tau$ a relation
$R_i(\strA)\subseteq V(\strA)^{\ar(R_i)}$, and for each $c_i\in\tau$ a
constant $c_i(\strA)\in V(\strA)$. If $\strA$ is a $\tau$-structure and $R$
is a relation symbol not in $\tau$ with associated arity $r$ and
$R(\strA)\subseteq V(\strA)^r$ is an $r$-ary relation over $V(\strA)$, we
write $(\strA,R(\strA))$ for the $\tau\cup\{R\}$-structure obtained by
extending $\strA$ with the relation $R(\strA)$. The \emph{order $|\strA|$} of a
$\tau$-structure $\strA$ is $|V(\strA)|$, and its \emph{size
  $\norm{\strA}$} is $|\tau|+|V(\strA)|+\sum_{R\in\tau}|R(\strA)|$, which
corresponds to the size of a representation of~$\strA$ in an appropriate
model of computation. We call a structure $G$ of signature $\{E\}$, where
$E$ is a binary relation symbol, a \emph{digraph}, and if $E(G)$ is symmetric and
irreflexive, we call $G$ a \emph{graph}. Let $V$ be a set. A \emph{successor
  relation} on $V$ is a binary relation $S\subseteq V\times V$ such that
$(V,S)$ is a directed path of length $|V|-1$. We write $\bar{a}$ for a
finite sequence $(a_1,\ldots,a_k)$ and usually leave it to the context to
determine the length of a sequence. The \emph{Gaifman-graph
  $G(\strA)$} of a $\tau$-structure $\strA$ is the graph with vertex set
$V(\strA)$ and edge set $\{(u,v)\mid\text{$u\neq v$ and there is an
  $R\in\tau$ and a tuple $\bar{a}\in R(\strA)$ such that
  $u,v\in\bar{a}$}\}$.

\paragraph*{First-order logic.}
We assume familiarity with first-order logic \FO and monadic second-order
logic \MSO. We write $\FO(\tau)$ and $\MSO(\tau)$ for the set of all \FO
and \MSO formulas over signature~$\tau$, respectively. If $\phi$ is a
formula of \FO or \MSO, we write $|\phi|$ for the length (of an encoding) of~$\phi$. If $\phi$ is a sentence of $\FO(\tau)$ or $\MSO(\tau)$ and $\strA$
a $\tau$-structure, we write $\strA\models\phi$ if $\phi$ is true 
in~$\strA$. If~$\phi(\bar{x})$ has free variables $\bar{x}$ and
$\bar{a}\in V(\strA)^k$ is a tuple of the same length as $\bar{x}$, we
write $\strA\models\phi(\bar{a})$ if $\phi$ is true in $\strA$, where the
free variables $\bar{x}$ are interpreted by the elements of $\bar{a}$ in
the obvious way. We write $\phi(\strA)$ for the relation
$R\coloneqq\{\bar{a}\mid\strA\models\phi(\bar{a})\}$. We call a formula
$\phi(\bar{x})$ over signature $\tau=\sigma\cup\{<\}$
\emph{order-invariant} if for every $\sigma$-structure $\strA$ and all linear orders $<_1,<_2$ over $V(\strA)$ we
have $\phi(\str A,<_1)=\phi(\str A,<_2)$. Analogously, we call a
formula~$\phi(\bar{x})$ over signature  $\tau=\sigma\cup\{S\}$ (where~$S$
is a binary relation symbol) \emph{successor-invariant} if for every
$\sigma$-structure $\strA$ and all
successor relations $S_1,S_2$ over $V(\strA)$ we have
$\phi(\strA,S_1)=\phi(\strA,S_2)$. We write $\FOoi$ and $\MSOoi$ for the set
of all order-invariant \FO and \MSO formulas, respectively, and $\FOsi$ and
$\MSOsi$ for the set of all successor invariant \FO and \MSO formulas,
respectively. We write $\FO[{<}]$ and $\MSO[{<}]$ for the set of all \FO
and \MSO formulas, respectively, over a signature which contains at least
the binary relation symbol $<$, and similarly for $\FO[{+1}]$ and
$\MSO[{+1}]$. For any order-invariant formula $\phi$ and any
$\tau$-structure $\strA$, we denote $\strA\models_{\mathrm{ord-inv}}\phi$
if for some (equivalently, every) order relation $<$ on the universe of
$\strA$ it holds that $(\strA,<)\models\phi$. Similarly, For a
successor-invariant formula $\phi$ and any $\tau$-structure $\strA$, we
denote $\strA\models_{\mathrm{succ-inv}}\phi$ if for some (equivalently,
every) successor relation $S$ on the universe of $\strA$ it holds that
$(\strA,S)\models\phi$.

Throughout the paper we study the complexity of order- and
successor-invariant logics on restricted classes of structures. As usual in
this type of research we focus on classes of graphs. More general
structures can be reduced to this case using their Gaifman-graph\ptxt{s}. In our
analysis we will use the framework of \emph{parameterised complexity},
see, e.g.\@ \cite{CyganFKLMPPS15, DowneyF99, FlumG06}. A
\emph{parameterised problem} is a subset of $\Sigma^*\times \N$ where
$\Sigma$ is a fixed finite alphabet. For an instance $(w,k) \in
\Sigma^*\times \N$, $k$ is called the parameter.

Let $\CCC$ be a class of graphs and $\LLL$ be one of first-order or monadic
second-order logic. The \emph{order-invariant model-checking problem}
$\MC(\LLL[{<}\textit{-inv}],\CCC)$ of the logic $\LLL$ on the class $\CCC$
of graphs is defined as the problem
\smallskip
\npprob{9.1cm}{$\MC(\LLL[{<}\textit{-inv}],\CCC)$}%
{$G\in\CCC$, order-invariant $\phi\in\LLL(\{E,<\})$}%
{$|\phi|$}%
{$G\models_{\mathrm{ord-inv}}\phi$\,?}

\smallskip
\noindent
Analogously, we define the \emph{successor-invariant model-checking
  problem} $\MC(\LLL[+1\textit{-inv}],\CCC)$ for~$\LLL$ on~$\CCC$, where
instead the formula $\phi \in \LLL(\{E,S\})$ is required to be successor-invariant. 

Finally,
we define the \emph{ordered model-checking problem} $\MC(\LLL[{<}],\CCC)$
for $\LLL$ on $\CCC$ as

\smallskip
\npprob{9.1cm}{$\MC(\LLL[{<}], \CCC)$}%
{$G\in\CCC$, $<$ a linear order of $V(G)$ and $\phi\in\LLL(\{E,<\})$}%
{$|\phi|$}%
{$(G, <)\models\phi$?}

\smallskip
\noindent
Likewise, we define the \emph{model-checking problem}
$\MC(\LLL[{+1}], \CCC)$ \emph{with successor} for~$\LLL$ on~$\CCC$, which
gets as input a graph $G\in\CCC$, a successor relation $S$ on $V(G)$, and a
formula $\phi\in \LLL(\{E,S\})$.
%

The order- and successor-invariant model-checking problems are
\emph{fixed-parameter tractable}, or in the complexity class \FPT, if there
are algorithms that correctly decide on input $(G,\phi)$  whether
${(G,<)}\models\phi$ for some linear order $<$ or $(G,S)\models\phi$ for
some successor relation $S$, respectively, in time
$f(|\phi|)\cdot\norm{G^{\mathcal{O}(1)}}$, for some function
$f:\N\to \N$ in the case where $\phi$ is order invariant or successor invariant, respectively. We have a similar definition of \FTP{} for
$\MC(\LLL[{<}])$ and $\MC(\LLL[{+1}])$. The model-checking problem for
first-order logic is complete for the parameterised complexity class
$\AWs$, which is conjectured and widely believed to strictly contain the
class \FPT. Thus it is widely believed that model-checking for first-order
logic (and thus also for monadic second-order logic) is not fixed-parameter
tractable.


\section{Model-Checking on Ordered Structures}\label{sec:order}

In this section we investigate the tractability of model-checking on
classes of ordered structures and on classes of structures with a successor relation.
Of course, it is easy to come up with classes of ordered structures on
which model-checking is fixed-parameter tractable, e.g.\ by taking the
class of all cliques with a linear order on the vertex set. Thus we seek
restrictions as weak as possible while still allowing us to show that
model-checking is \AWs-hard.

\subsection{Coloured Sets}

We start by observing that on the class of ordered coloured sets (and, a
forteriori, on the class of coloured sets with a successor relation), model-checking is tractable even for monadic second-order logic. This is
B\"uchi-Elgot-Trakhtenbrot's Theorem (cf.~\cite{FlumG06}), since coloured
ordered sets are just strings. Thus model-checking for \MSO is
fixed-parameter tractable on structures whose signature contains only unary
relation symbols, apart from the order relation.

\subsection{Vertex-Ordered Graphs}

The simplest case not covered by the preceding paragraph is that of ordered
graphs, i.e.\ $\{E, {<}\}$-structures where both $E$ and ${<}$ are binary
relation symbols. We show that model-checking for first-order logic is
\AWs-hard even for very simple graphs.

\begin{theorem}\label{lem:MCForest}
  Let $\CCC$ be a class of graphs. If $\CCC$ contains all partial
  matchings, then $\MC(\FOo,\CCC)$ is \AWs-hard. If $\CCC$ contains all
  star forests, then $\MC(\FOs,\CCC)$ is \AWs-hard.
\end{theorem}
Here, a \emph{partial matching} is a disjoint union of edges and
  isolated vertices (i.e.\@ a graph of maximum degree~1), while a
\emph{star forest} is a disjoint union of stars (complete bipartite
  graphs $K_{1,n}$ with $n\geq0$). Note that on both these graph
  classes, the model-checking problem for plain \FO{} is fixed-parameter tractable.

\begin{proof}
  For the first part we show how to construct in polynomial time for every
  graph $G$ a $\{E,{<}\}$-structure $\strA$ such that $G$ can be
  \FO-interpreted in $\strA$. For this, let $\{v_1,\ldots,v_n\}$ be the
  vertex set of $G$ ordered in an arbitrary way, and assume that $v_i$ has
  degree $d_i$ in $G$. To each vertex in $G$ we associate an interval of
  length ${\hat d}_i\coloneqq\max\{d_i,1\}$ in $\strA$, and separate
  the intervals by gaps of length $2$. Thus with
  \[D_k\coloneqq 2k-1+\sum_{i=1}^{k-1}{\hat d}_i,\]
  we associate with $v_i$ the interval
  $\{D_i,\ldots,D_i+{\hat d}_i-1\}$. The edge set $E(\strA)$ consists of
  the edges $\{D_i-2,D_i-1\}$ for $i\geq2$, together with the edges
  $\{D_i+k,D_j+\ell\}$ if $v_iv_j$ is an edge of $G$, $v_j$ is the $k$-th
  neighbour of $v_i$ in the ordering, and $v_i$ is the $\ell$-th neighbour of $v_i$. Notice
  that the edges $\{D_i-2,D_i-1\}$ are the only edges between consecutive
  elements, so they can be used to determine the intervals used in this
  construction.

  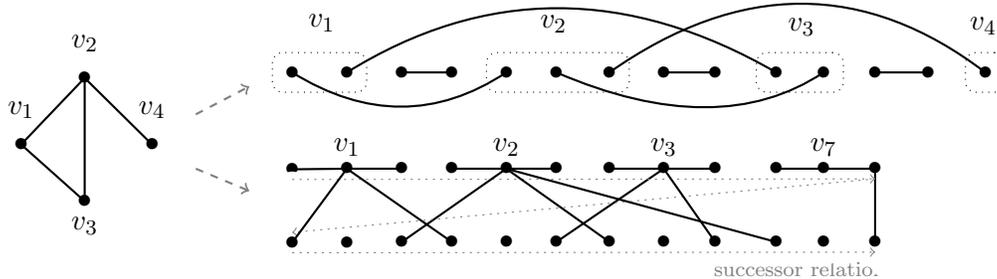
\begin{figure}[htb]
    \begin{center}
    \begin{tikzpicture}[scale=0.3,auto]
        \tikzstyle{every path}=[thick]
        \tikzstyle{every node}=[circle,fill=white,text=black,thin]
        \tikzset{quadratic bezier/.style={ to path={(\tikztostart) .. controls($#1!1/3!(\tikztostart)$)
        and ($#1!1/3!(\tikztotarget)$).. (\tikztotarget)}}}

        \node  at (18cm,-16.2cm) {\textcolor{gray}{\scriptsize successor relation}};        
        \node[label={[label distance=-12pt]$v_1$}]  (n0) at (-2.0691307803897936cm,-11.769245473472093cm) {$\bullet$};
        
        \node[label={[label distance=-12pt]$v_2$}]  (n1) at (4.994188199305022cm,-11.769245473472093cm) {$\bullet$};
        
        \node  (n2) at (7.210863202033742cm,-11.769245473472093cm) {$\bullet$};
        
        \node[label={[label distance=-12pt]$v_3$}]  (n3) at (11.973460164741198cm,-11.769245473472093cm) {$\bullet$};
        
        \node[inner sep=0pt]  (n4) at (-4.492627531059064cm,-11.795703806805426cm) {$\bullet$};
        
        \node  (n5) at (0.36155279667794105cm,-11.769245473472093cm) {$\bullet$};
        
        \node  (n6) at (2.5817722686654747cm,-11.769245473472093cm) {$\bullet$};
        
        \node[label={[label distance=-12pt]$v_7$}]  (n7) at (19.053769581098052cm,-11.769245473472093cm) {$\bullet$};
        
        \node  (n8) at (9.564342135128301cm,-11.769245473472093cm) {$\bullet$};
        
        \node[inner sep=0pt]  (n9) at (21.33888369308661cm,-11.769245473472093cm) {$\bullet$};
        
        \node  (n10) at (14.250244219610188cm,-11.769245473472093cm) {$\bullet$};
        
        \node  (n11) at (16.952706454167114cm,-11.769245473472093cm) {$\bullet$};
        
        \node  (n12) at (11.973460164741198cm,-15.018203806805426cm) {$\bullet$};
        
        \node  (n13) at (0.36155279667794105cm,-15.018203806805426cm) {$\bullet$};
        
        \node[inner sep=0pt]  (n14) at (-4.492627531059064cm,-15.044662140138758cm) {$\bullet$};
        
        \node  (n15) at (-2.0691307803897936cm,-15.044662140138758cm) {$\bullet$};
        
        \node  (n16) at (7.210863202033742cm,-15.018203806805426cm) {$\bullet$};
        
        \node  (n17) at (14.250244219610188cm,-15.018203806805426cm) {$\bullet$};
        
        \node  (n18) at (16.952706454167114cm,-15.018203806805426cm) {$\bullet$};
        
        \node  (n19) at (2.5817722686654747cm,-15.018203806805426cm) {$\bullet$};
        
        \node  (n20) at (9.564342135128301cm,-15.018203806805426cm) {$\bullet$};
        
        \node[inner sep=0pt]  (n21) at (21.33888369308661cm,-15.018203806805426cm) {$\bullet$};
        
        \node  (n22) at (4.994188199305022cm,-15.018203806805426cm) {$\bullet$};
        
        \node  (n23) at (19.053769581098052cm,-15.018203806805426cm) {$\bullet$};
        
        \node[label={[label distance=-7pt]$v_1$}]  (n24) at (-16.505566661790134cm,-10.69534743635097cm) {$\bullet$};
        
        \node[label={[label distance=-7pt]$v_2$}]  (n25) at (-13.68862447942709cm,-7.744265150065875cm) {$\bullet$};
        
        \node[label={[label distance=-7pt,below]$v_3$}]  (n26) at (-13.68862447942709cm,-13.2104743848894cm) {$\bullet$};
        
        \node[label={[label distance=-7pt]$v_4$}]  (n27) at (-10.704007167161485cm,-10.69534743635097cm) {$\bullet$};
        
        \node[inner sep=0pt]  (n28) at (-4.492627531059064cm,-7.498474649425596cm) {$\bullet$};
        
        \node[inner sep=0pt]  (n29) at (-2.0691307803897936cm,-7.498474649425596cm) {$\bullet$};
        
        \node  (n30) at (0.36155279667794105cm,-7.498474649425596cm) {$\bullet$};
        
        \node  (n31) at (2.5817722686654747cm,-7.498474649425596cm) {$\bullet$};
        
        \node[inner sep=0pt]  (n32) at (4.994188199305022cm,-7.498474649425596cm) {$\bullet$};
        
        \node  (n33) at (7.210863202033742cm,-7.498474649425596cm) {$\bullet$};
        
        \node[inner sep=0pt]  (n34) at (9.564342135128301cm,-7.498474649425596cm) {$\bullet$};
        
        \node  (n35) at (11.973460164741198cm,-7.498474649425596cm) {$\bullet$};
        
        \node  (n36) at (14.250244219610188cm,-7.498474649425596cm) {$\bullet$};
        
        \node[inner sep=0pt]  (n37) at (16.952706454167114cm,-7.498474649425596cm) {$\bullet$};
        
        \node[inner sep=0pt]  (n38) at (19.053769581098052cm,-7.498474649425596cm) {$\bullet$};
        
        \node  (n39) at (21.33888369308661cm,-7.498474649425596cm) {$\bullet$};
        
        \node  (n40) at (23.690598550876693cm,-7.498474649425596cm) {$\bullet$};
        
        \node[inner sep=0pt]  (n41) at (26.22000641648959cm,-7.498474649425596cm) {$\bullet$};

        \node  (n42) at (-3.1861626058479464cm,-5.296910567502149cm) {$v_1$};

        \node  (n43) at (7.0902849032440205cm,-5.341397353255793cm) {$v_2$};

        \node  (n44) at (18.078520984394302cm,-5.474857710516726cm) {$v_3$};

        \node  (n45) at (26.130629205804027cm,-5.430370924763082cm) {$v_4$};

        \node  (n46) at (-8.763022383889018cm,-9.357527203252673cm) {};
        \node  (n47) at (-6.395281907109553cm,-8.122184345802516cm) {};
        \node  (n48) at (-8.763022383889018cm,-11.769245473472093cm) {};
        \node  (n49) at (-6.395281907109553cm,-12.754720061240597cm) {};
        \node  (n50) at (-4.492627531059064cm,-12.5cm) {};
        \node  (n53) at (21.803990707406722cm,-12.775507503129784cm) {};
        \node  (n54) at (-4.8085550042430825cm,-13.931876055731788cm) {};
        \node  (n55) at (-5.862677687314353cm,-15.018203806805426cm) {};
        \node  (n56) at (-4.8085550042430825cm,-16.094367609849936cm) {};
        \node  (n57) at (21.803990707406722cm,-16.014367498440926cm) {};
	\node[draw,dotted,rectangle,rounded corners,fill=none,fit=(n28)(n29)]{};
	\node[draw,dotted,rectangle,rounded corners,fill=none,fit=(n32)(n34)]{};
	\node[draw,dotted,rectangle,rounded corners,fill=none,fit=(n37)(n38)]{};
	\node[draw,dotted,rectangle,rounded corners,fill=none,fit=(n41)]{};

        \draw (n9.center) to (n21.center);
        \draw (n8.center) to (n3.center);
        \draw (n25.center) to (n27.center);
        \draw (n39.center) to (n40.center);
        \coordinate (crtl33)  at (14.398750000000003,-10.588750000000003);
        \draw (n33.center) .. controls($(crtl33)!1/3!(n33)$) and ($(crtl33)!1/3!(n38)$).. (n38.center);
        \draw[dashed,gray,->] (n46.center) to (n47.center);
        \coordinate (crtl28)  at (0.4022916666666658,-10.535833333333336);
        \draw (n28.center) .. controls($(crtl28)!1/3!(n28)$) and ($(crtl28)!1/3!(n32)$).. (n32.center);
        \draw (n3.center) to (n16.center);
        \draw (n3.center) to (n10.center);
        \draw[dotted,gray,->,thin] (n4.south) to (n9.south);
        \coordinate (crtl34)  at (18.235208333333333,-1.4077083333333338);
        \draw (n34.center) .. controls($(crtl34)!1/3!(n34)$) and ($(crtl34)!1/3!(n41)$).. (n41.center);
        \coordinate (crtl29)  at (7.0962499999999995,-1.8045833333333352);
        \draw (n29.center) .. controls($(crtl29)!1/3!(n29)$) and ($(crtl29)!1/3!(n37)$).. (n37.center);
        \draw (n1.center) to (n18.center);
        \draw (n1.center) to (n13.center);
        \draw[dotted,gray,->,thin] (n9.south) to (n14.north);
        \draw[dashed,gray,->] (n48.center) to (n49.center);
        \draw (n3.center) to (n17.center);
        \draw (n35.center) to (n36.center);
        \draw (n25.center) to (n26.center);
        \draw (n1.center) to (n20.center);
        \draw (n24.center) to (n25.center);
        \draw (n26.center) to (n24.center);
        \draw (n7.center) to (n9.center);
        \draw (n0.center) to (n14.center);
        \draw (n6.center) to (n1.center);
        \draw (n0.center) to (n19.center);
        \draw[dotted,gray,->,thin] (n14.south) to (n21.south);
        \draw (n4.center) to (n0.center);
        \draw (n30.center) to (n31.center);
        \draw (n1.center) to (n2.center);
        \draw (n0.center) to (n5.center);
        \draw (n11.center) to (n7.center);
    \end{tikzpicture}
    \end{center}
    \caption{A sample graph (left) encoded in a linear order plus perfect
      matching (upper right) and a star forest plus successor relation
      (lower right).}
    \label{fig:foo}
  \end{figure}
  
  For the second part we construct a structure $\strA'$ consisting of a
  disjoint union of stars and a successor relation that can be used to
  recover the original graph using an \FO interpretation. Again, we assume
  the vertex set of $G$ to be $\{v_1,\ldots,v_n\}$. A vertex $v$ is
  encoded by a path $v^{-1}, v, v^{+1}$. The vertices of these 
  paths are placed at   the beginning of the successor relation in an arbitrary order. An edge $e$ is encoded by 
  three vertices $e^{-1}, e, e^{+1}$ such that $e$ is a direct
  successor of $e^{-1}$ and $e^{+1}$ is a direct successor of $e$. All
  these vertices are placed at the end of the successor relation. 
  For every edge $e = vw$, assume that $v$ is smaller than $w$ in 
  the successor relation. We connect, in $\strA'$,
  $v$ to $e^{-1}$ and $w$ to $e^{+1}$. Again, $G$ may be 
  recovered from $\strA'$ using an \FO  interpretation.
\end{proof}

As a corollary of the previous theorem we get that $\MC(\FOs,\mathscr{P})$ and
$\MC(\FOo,\mathscr{T})$ are \AWs-hard for the class $\mathscr{P}$ of planar graphs and the class $\mathscr{T}$ of graphs of treewidth $1$ (forests). However, for $\MC(\FOs,\CCC)$ to be \AWs-hard
it is essential that the graphs in the class $\mathscr{T}$ have unbounded degree.
Indeed, on graph classes of bounded degree, successor-invariant \FO
model-checking is tractable.

\begin{theorem}
  For every $d\geq0$ let $\Cdegd$ be the class of graphs of maximum degree
  at most $d$. Then for all $d\geq0$, $\MC(\FOs,\Cdegd)$ is fixed-parameter
  tractable. In fact, we can allow any (fixed) number of successor
  relations on top of $\DDD_d$ and still have tractable first-order
  model-checking.
\end{theorem}

\begin{proof}
  By a result of Seese \cite{Seese96} the model-checking problem for \FO on
  graphs of bounded degree and also on all structures with Gaifman-graph of
  bounded degree is fixed-parameter tractable. Adding a successor relation
  increases the degree of the Gaifman-graph of a structure by at most two.
\end{proof}

\section{Order-Invariant MSO}
\label{sec:order-inv}

In this section we consider order-invariant logics. The most expressive
logic studied in the context of algorithmic meta-theorems is monadic
second-order logic, the extension of first-order logic by quantification
over sets of elements. With respect to graphs, there are two variants of
\MSO usually considered, one, called $\MSO_1$, where we can quantify over
sets of vertices, and the other, called $\MSO_2$, where we can additionally
quantify over sets of edges. It was shown by Courcelle~\cite{Courcelle90}
that $\MSO_2$ is fixed-parameter tractable on every class of graphs of
bounded treewidth. Later, Courcelle et al.~\cite{CourcelleMakRot00} showed
that $\MSO_1$ is fixed-parameter tractable on every class of graphs of
bounded clique-width, a concept more general than bounded treewidth. In this
section, we show that for both logics we can allow order-invariance without
increase in complexity.

We first consider the case of $\MSO_2$. As shown in Theorem 5.1.1.\ of~\cite{quiroz2017chromatic}, for every graph $G$ of treewidth
$k$ there is a successor relation $S$ on $V(G)$ such that the graph
obtained from $G$ by adding the edges in~$S$ has treewidth at most $k+2$. From the proof one can easily derive an algorithm that 
takes as input a a graph $G$ and a tree decomposition of $G$ 
of width $k$ and outputs a successor relation as desired in polynomial
time. We can use the algorithm of 
Bodlaender~\cite{bodlaender1996linear} to compute an 
optimal tree decomposition in time 
$2^{\mathcal{O}(k^3)}\cdot n$ first, and then 
compute the desired successor relation. We also refer to 
\cite{Makowsky05} and \cite{ChenF12} for proofs that 
for every graph $G$ of treewidth~$k$ there is a successor relation 
$S$ on $V(G)$ such that the graph
obtained from $G$ by adding the edges in $S$ has treewidth at most $k+5$.
In combination with Courcelle's theorem, this
implies the following result.

\begin{theorem}\label{theo:mso2-oi}
  For any class $\CCC$ of bounded treewidth, $\MC(\MSOtwooi,\CCC)$ is
    fixed-parameter tractable.
\end{theorem}

In fact, $\MC(\MSOtwooi)$ is fixed-parameter tractable with parameter
$|\phi|+\tw(G)$, where $\tw(G)$ is the treewidth of a graph $G$. We prove next that also for $\MSO_1$ and clique-width we can
allow order-invariance without loss of tractability.

\begin{theorem}\label{theo:mso-oi}
  For any class $\CCC$ of graphs of bounded clique-width,
    $\MC(\MSOoneoi,\CCC)$ is fixed-parameter tractable.
\end{theorem}

We first review the definition of clique-width. For the rest of this section
we fix a relational signature $\sigma$ in which every relation symbol has
arity at most $2$.

\begin{definition}[$\sigma$-clique-expression of width $k$]
  Let $k\in\N$ be fixed. A $\sigma$-clique-expression of width~$k$ is a
  pair $(T,\lambda)$, where $T$ is a directed rooted tree in which
    all edges are directed away from the root and
  \begin{align*}
    \lambda:V(T)\to{} &\{\mathbf{1},\dots,\mathbf{k}\}\cup
    \{\oplus\}\\
    &\cup \{\textup{edge}_{R,i\to j}\mid R\in\sigma,\:
    i,j=1,\ldots,k\}\cup \{\textup{rename}_{i\to j}\mid i,j=1,\ldots,k\},
  \end{align*}
  such that for all $t\in V(T)$ we have: if
  $\lambda(t)\in\{\mathbf{1},\dots,\mathbf{k}\}$, then $t$ is a leaf of
  $T$; if $\lambda(t)=\oplus$, then $t$ has exactly two successors; and in
  all other cases $t$ has exactly one successor.
\end{definition}  

\pagebreak
\begin{definition}
  Let $(T,\lambda)$ be a $\sigma$-clique-expression of width $k$. With
  every $t\in V(T)$ we associate a $\sigma$-structure $G(t)$ in which
  vertices are coloured by colours $\mathbf{1},\dots,\mathbf{k}$ as
  follows.
  \begin{itemize}
  \item If $t$ is a leaf, then $G(t)$ consists of one element coloured by
    $\lambda(t)$.
  \item If $\lambda(t)=\oplus$ and $t$ has successors $t_1,t_2$, then
    $G(t)$ is the disjoint union $G(t_1)\mathop{\dot\cup}G(t_2)$.
  \item If $\lambda(t)=\textup{edge}_{R,i\to j}$ and $t_1$ is the successor
    of $t$, then $G(t)$ is the structure obtained from $G(t_1)$ by adding
    to the relation $R(G(t))$ all pairs $(u,v)$ such that $u$ has colour
    $i$ and $v$ has colour $j$.
  \item If $\lambda(t)=\textup{rename}_{i\to j}$ and $t_1$ is the successor
    of $t$, then $G(t)$ is the structure obtained from~$G(t_1)$ by changing
    the colour of all vertices $v$ which have colour $i$ in $G(t_1)$ to
    colour $j$ in~$G(t)$.
  \end{itemize}
  The \emph{$\sigma$-structure generated by $(T, \lambda)$} is the
  structure $G(r)$, where $r$ is the root of $T$, from which we remove
  all colours $\{ \mathbf{1}, \dots, \mathbf{k}\}$.  Finally, the
  \emph{clique-width} of a $\sigma$-structure $G$ is the minimal width
  of a clique-expression generating $G$.
\end{definition}

Combining results from \cite{HlinenyO08} and \cite{OumS06} yields the
following well-known result.

\begin{theorem}\label{theo:comp-clique}
  There is a computable function $f:\N\to\N$ and an algorithm which, given
  a graph~$G$ of clique-width at most $k$ as input, computes a
  clique-expression of width at most $2^{k+1}$ in time $f(k)\cdot|G|^3$.
\end{theorem}

Combining this theorem with results of 
 Courcelle et al.~\cite{CourcelleMakRot00} 
yields the following (also well-known) result.
 
\begin{theorem}\label{theo:mso-clique}
  For any class $\CCC$ of graphs of bounded clique-width,
    $\MC(\MSO_1,\CCC)$ is fixed-parameter tractable.
\end{theorem}

In fact, the result applies to any $\sigma$-structure of bounded
clique-width, provided that a clique-expression generating the structure (whose width is bounded by a computable function of the clique-width of the structure)
is given or computable in polynomial time.

\medskip
The next lemma is the main technical ingredient for the proof of
  Theorem~\ref{theo:mso-oi} above.

\begin{lemma}\label{lem:order-clique}
  There is an algorithm which, on input a graph $G$ of clique-width at most
  $k$, computes a linear order $<$ on $V(G)$ and a clique-expression of
  width at most $2^{k+2}$ generating the structure~$(G, <)$.
\end{lemma}

\begin{proof}
  Let $G$ and $k$ be given. Using Theorem~\ref{theo:comp-clique} we first
  compute an $\{E\}$-clique-expression~$(T,\lambda)$ of width at most
  $2^{k+1}$ generating $G$. Let $r$ be the root of $T$. For every node
  $t\in V(T)$ we fix an ordering of its successors. Let $\prec$ be the
  partial order on $V(T)$ induced by this ordering.

  Let $t\in V(T)$ be a node and let $s\ne t$ be the first node on the path
  $P$ from $t$ to $r$ with $\lambda(s)=\oplus$, if it exists. Let $t_1,t_2$
  be the successors of $s$ with $t_1\prec t_2$. We call $t$ a \emph{left
    node} if $t_1\in V(P)$, and a \emph{right node} otherwise. If there is
  no node labelled $\oplus$ strictly above $t$, then we call $t$ a left
  node as well.

  For every $t\in V(T)$ let $T_t$ be the subtree of $T$ with root $t$, and
  let $\lambda_{|T_t}$ be the restriction of $\lambda$ to the subtree
  $T_t$. We recursively define a transformation $\rho(T_t,\lambda_{|T_t})$
  on the subtrees of $T$ defined as follows. Intuitively, we produce a new clique-expression $(T',\lambda')$ over the signature
  $\{E,<\}$ using colours $\{(i,a), (i, b)\mid 1\leq i\leq k\}$.
  Essentially, the new clique-expression will generate the same graph as
  $(T, \lambda)$, but so that if $t$ is a node in $T$ and $T_t$ generates
  the graph $G_t$, then $T'$ contains a node~$t'$ generating an ordered
  version $G'_t \coloneqq (G_t, <)$ of $G_t$ so that if $v\in V(G_t)$ has
  colour $i$, then, in~$G'_t$, $v$ has colour $(i,a)$ if $t$ is a left node
  and $(i,b)$ if $t$ is a right node. Hence, whenever in $T$ we take the
  disjoint union of $G_t$ and $G_s$ and $t\prec s$, then we can define the
  ordering on $G'_t\mathop{\dot\cup}G'_s$ by adding all edges from nodes in
  $G'_t$ to $G'_s$, i.e.\ all edges from vertices coloured $(i,a)$ to
  $(j,b)$ for all pairs $i,j$.

  Formally, the transformation is defined as follows.
  \begin{itemize}
  \item If $t\in V(T)$ is a leaf, then $\rho(t)\coloneqq(T',\lambda')$,
    where $T'$ consists only of $t$ and
    $\lambda'(t)\coloneqq(\lambda(t),a)$ if $t$ is a left node, and
    $\lambda'(t)\coloneqq(\lambda(t),b)$ if $t$ is a right node.
  \item Suppose $\lambda(t)=\textup{rename}_{i\to j}$ and let $s$ be the
    successor of $t$. Then
    $\rho(T_t,\lambda_{|T_t})\coloneqq(T',\lambda')$, where~$T'$ is a tree
    defined as follows. Let $(T'',\lambda'')\coloneqq
    \rho(T_s,\lambda_{|T_s})$ and let $r''$ be the root of $T''$. Then~$T'$
    is obtained from $T''$ by adding a new root $r'$, a new vertex $v$, and
    new edges $(r',v)$ and $(v,r'')$. We define
    $\lambda'(r')\coloneqq\textup{rename}_{(i,a)\to(j,a)}$, 
    $\lambda'(v)\coloneqq\textup{rename}_{(i,b)\to(j,b)}$, and
    $\lambda'(u)\coloneqq\lambda''(u)$ for $u\in V(T'')$.
  \item Suppose $\lambda(t)=\textup{edge}_{E,i\to j}$ and let $s$ be the
    successor of $t$. Then
    $\rho(T_t,\lambda_{|T_t})\coloneqq(T',\lambda')$, where~$T'$ is a tree
    defined as follows. Let
    $(T'',\lambda'')\coloneqq\rho(T_s,\lambda_{|T_s})$ and let $r''$ be the
    root of $T''$. Then~$T'$ is obtained from $T''$ by adding a path
    $(v_1,v_2,v_3,v_4)$ of length $3$ and making $r''$ a successor
    of~$v_4$. We define
    $\lambda'(v_1)\coloneqq\textup{edge}_{E,(i,a)\to(j,a)}$,
    $\lambda'(v_2)\coloneqq\textup{edge}_{E,(i,b)\to(j,a)}$,
    $\lambda'(v_3)\coloneqq\textup{edge}_{E,(i,a)\to(j,b)}$,
    $\lambda'(v_4)\coloneqq\textup{edge}_{E,(i,b)\to(j,b)}$, and
    $\lambda'(u)\coloneqq\lambda''(u)$ for $u\in V(T'')$.
  \item Finally, suppose $\lambda(t)=\oplus$ and let $t_1,t_2$ be the
    successors of $t$ such that $t_1\prec t_2$. Then
    $\rho(T_t,\lambda_{|T_t})\coloneqq(T',\lambda')$, where $T'$ is a tree
    defined as follows. For $i=1,2$, let
    $(T_i,\lambda_i)\coloneqq\rho(T_{t_i},\lambda_{|T_{t_i}})$ and let
    $r_i$ be the root of $T_i$. $T'$ consists of the union of $T_1$
    and $T_2$, and the additional vertices $v_i$,
    $1\leq i\leq k$, $w_{i,j}$, $1\leq i,j\leq k$, and $v_{\oplus}$. We add
    the edges $(v_i,v_{i+1})$ for $1\leq i<k$, the edges
    $(w_{i,j},w_{i,j+1})$ for $1\leq i\leq k$, $1\leq j<k$, and
    the edges $(w_{i,k},w_{i+1,1})$ for $1\leq i<k$, and
    finally the edges $(v_k,w_{1,1})$, $(w_{k,k},v_\oplus)$, and
    $(v_\oplus,r_i)$ for $i=1,2$. For every node $s\in V(T_i)$ we define
    $\lambda'(s)\coloneqq\lambda_i(s)$, $i=1,2$. Furthermore, we define
    $\lambda(v_\oplus)\coloneqq\oplus$ and
    $\lambda'(w_{i,j})\coloneqq\textup{edge}_{<,(i,a)\to(j,b)}$ for
    $1\leq i,j\leq k$. Finally, if $t$ is a left node, then we
    define $\lambda(v_i)\coloneqq\textup{rename}_{(i,b)\to(i,a)}$ for
    $i\leq k$, and if $t$ is a right node, then we define
    $\lambda(v_i)\coloneqq\textup{rename}_{(i,a)\to(i,b)}$ for
      $i\leq k$.
  \end{itemize}

  Now, it is easily seen that $(T',\lambda')$ generates an
  $\{E,<\}$-structure $(V, E,<)$ where $(V, E)$ is the graph generated by
  $(T, \lambda)$ and $<$ is a linear order on $V$. The width of
  $(T',\lambda')$ is twice the width of $(T,\lambda)$, and hence at most
  $2^{k+2}$.
\end{proof}

We are now ready to prove Theorem~\ref{theo:mso-oi}.

\begin{proof}[Proof of Theorem~\ref{theo:mso-oi}] \
  Let $\CCC$ be a class of graphs of clique-width at most $k$. On input
  $G\in\CCC$ and $\phi\in\MSOoneoi$, we apply
  Lemma~\ref{lem:order-clique} to obtain a clique-expression $(T,\lambda)$
  of width $2^{k+2}$ generating an ordered copy $(G,<)$ of $G$ in
    time $f(k)\cdot|G|^3$, for some computable function $f$. We can now
  apply Theorem~\ref{theo:mso-clique} to decide whether $(G,<)\models\phi$
  in time $g(2^{k+2})\cdot p(|G|)$, where $g$ is a computable
  function and $p$ a polynomial. As $\phi$ is order-invariant, if
  $(G,<)\models\phi$, then $(G,<')\models\phi$ for any linear order $<'$ on
  $G$. Hence, if $(G,<)\models \phi$ we can return ``yes'' and otherwise
  reject the input. This concludes the proof.
\end{proof}

It is worth pointing out the following feature of the model-checking
algorithms established in Theorems~\ref{theo:mso2-oi}
and~\ref{theo:mso-oi}. Instead of designing new model-checking
algorithms, we reduce the verification of order-invariant \MSO on classes
of small treewidth or clique-width to the standard model-checking algorithms
for \MSO on classes of (slightly larger) treewidth and clique-width,
respectively. The advantage of this approach is that we can reuse existing
results on \MSO on such classes of graphs. For instance, in
\cite{KneisLR11} the authors report on a practical implementation of
Courcelle's theorem, i.e.\ on the implementation of a model-checker for
$\MSO_2$ on graph classes of bounded treewidth, and obtain astonishing
performance results in practical tests. Our technique allows us to reuse
this implementation so that with minimal effort it is possible to implement
our algorithm on top of the work in \cite{KneisLR11}.

Furthermore, in \cite{FlumFG02} it is shown that on graph classes $\CCC$ of
bounded treewidth, the set of all satisfying assignments of a given \MSO
formula $\phi(X)$ with free variables in a graph $G\in \CCC$ can be
computed in time linear in the size of the output and the size of $G$.
Again we can use the same algorithm to obtain the same result for
order-invariant \MSO.


\section{Successor-invariant FO on classes of bounded expansion}
\label{sec:mc-sparse}

Classes of bounded expansion are classes of uniformly sparse graphs that
have very good structural and algorithmic properties. Most notably, these
classes admit efficient first-order model-checking, as shown by
Dvo\v{r}\'ak et al.~\cite{dvovrak2013testing}. In this section we are going
to lift this result to successor-invariant formulas.
Let us give the required definitions first.

\paragraph*{Shallow minors and bounded expansion.}
A graph $H$ is a \emph{minor of $G$}, written $H\minor G$, if there are
pairwise disjoint connected subgraphs $(I_u)_{u\in V(H)}$ of $G$, called
\emph{branch sets}, such that whenever $uv\in E(H)$, then there are
$x_u\in I_u$ and $x_v\in I_v$ with $x_ux_v\in E(G)$. We call the family
$(I_u)_{u\in V(H)}$ a \emph{minor model} of $H$ in~$G$. A graph $H$ is a
\emph{depth-$r$ minor} of~$G$, denoted $H\minor_r G$, if there is a minor
model $(I_u)_{u\in V(H)}$ of~$H$ in $G$ such that each subgraph $I_u$ has
radius at most $r$.
For a graph $G$ and $r\in\N$, we write $\nabla_r(G)$ for the maximum edge
density $|E(H)|/|V(H)|$ of a graph $H\minor_r G$.

\begin{definition}
  A class of graphs $\CCC$ has \emph{bounded expansion} if there is a
  function $f:\N\to\N$ such that $\nabla_r(G)\leq f(r)$ for all
    $r\in\N$ and all $G\in\CCC$.
\end{definition}

Let $\tau$ be a finite and purely relational signature and let $\CCC$ be a
class of $\tau$-structures. We say that~$\CCC$ has \emph{bounded expansion}
if the class $\{G(\strA)\mid \strA\in\CCC\}$ of the Gaifman graphs of the
structures from~$\CCC$ has bounded expansion.

\paragraph*{Generalised colouring numbers.}
We will mainly rely on an alternative characterisation of bounded expansion
classes via \emph{generalised colouring numbers}. Let us fix a graph $G$.
By $\Pi(G)$ we denote the set of all linear orderings of $V(G)$. For
$L\in\Pi(G)$, we write $u<_L v$ if $u$ is smaller than $v$ in $L$, and
$u\le_L v$ if $u<_L v$ or $u=v$.

For $r\in\N$, we say that a vertex $u$ is \emph{strongly
  $r$-reachable} from a vertex~$v$ with respect to~$L$ if $u\le_L v$ and
there is a path $P$ of length at most $r$ that starts in $v$, ends in $u$,
and all whose internal vertices are larger than $v$ in $L$. By
$\SReach_r[G,L,v]$ we denote the set of vertices that are strongly
$r$-reachable from~$v$ with respect to $L$. Note that
$v\in\SReach_r[G,L,v]$ for any vertex~$v$.

We define the \emph{$r$-colouring number} of $G$ with respect to
  $L$ as
\[\col_r(G,L)=\max_{v\in V(G)}\bigl|\SReach_r[G,L,v]\bigr|,\]
and the \emph{$r$-colouring number} of $G$ (sometimes called
  \emph{strong $r$-colouring number}) as
\[\col_r(G)=\min_{L\in \Pi(G)} \col_r(G,L).\]

For $r\in\N$ and ordering $L\in\Pi(G)$, the \emph{$r$-admissibility}
$\adm_r[G,L,u]$ of a vertex~$v$ with respect to~$L$ is defined as the
maximum size of a family $\Pp$ of paths that satisfies the following two
properties:
\begin{itemize}
\item each path $P\in\Pp$ has length at most $r$, starts in $v$, and is
  either the trivial length-zero path or ends in a vertex $u<_L v$ and all
  its internal vertices are larger than~$v$ in $L$;
\item the paths in $\Pp$ are pairwise vertex-disjoint, apart from
  sharing the start vertex $v$.
\end{itemize}
The \emph{$r$-admissibility of $G$} with respect to $L$
is defined similarly to the $r$-colouring number:
\[\adm_r(G,L)=\max_{v\in V(G)} \adm_r[G,L,v],\]
and the \emph{$r$-admissibility of $G$} is given by
\[\adm_r(G)=\min_{L\in \Pi(G)} \adm_r(G,L).\]

The $r$-colouring numbers were introduced by Kierstead and Yang
\cite{kierstead2003orders}, while $r$-admissibility was first studied by
Dvo\v{r}\'ak~\cite{dvovrak13}. It was shown that those parameters are
related as follows.
\begin{lemma}[\rm Dvo\v{r}\'ak \cite{dvovrak13}]\label{lem:gen-col-ineq}
  For any graph $G$, $r\in\N$ and vertex ordering
  $L\in\Pi(G)$, we have
  \[\adm_r(G,L)\le \col_r(G,L)\le \bigl(\adm_r(G,L)\bigr)^r.\]
\end{lemma}
(Note that in Dvo\v{r}\'ak's work, the reachability sets do
  not include the starting vertex, hence the above inequality is stated
slightly differently in \cite{dvovrak13}.)

As proved by Zhu \cite{zhu2009coloring}, the generalised colouring numbers
are tightly related to densities of low-depth minors, and hence they can be
used to characterise classes of bounded expansion.

\begin{theorem}[\rm Zhu \cite{zhu2009coloring}]\label{thm:b_exp_b_deg}
  A class $\CCC$ of graphs has bounded expansion if and only if there is a
  function $f:\N\to\N$ such that $\col_r(G)\le f(r)$ for all $r\in\N$ and
  all $G\in\CCC$.
\end{theorem}

We need to be a bit more precise and use the following
lemma.

\begin{lemma}[\rm Grohe et al.~\cite{grohe2015colouring}]\label{lem:adm-nabla}
  For any graph $G$ and $r\in\N$ we have
  $\adm_r(G)\leq 6r\cdot(\nabla_r(G))^3$.
\end{lemma}

As shown by Dvo\v{r}\'ak \cite{dvovrak13}, on classes of bounded expansion
one can compute $\adm_r(G)$ in linear fixed-parameter time, parameterised by
$r$. More precisely, we have the following.

\begin{theorem}[\rm Dvo\v{r}\'ak \cite{dvovrak13}]\label{thm:adm-compute}
  Let $\CCC$ be a class of bounded expansion. Then there is an algorithm
  that, given a graph $G\in \CCC$ and $r\in\N$, computes a vertex
  ordering $L\in \Pi(G)$ with $\adm_r(G,L)=\adm_r(G)$ in time
  $f(r)\cdot |V(G)|$, for some computable function $f$.
\end{theorem}

We remark that Dvo\v{r}\'ak states the result in \cite{dvovrak13} as the
existence of a linear-time algorithm for each fixed value of~$r$. However,
an inspection of the proof reveals that it is actually a single
fixed-parameter algorithm that can take $r$ as input. To the best of our
knowledge, a similar result for computing $\col_r(G)$ is not known, but by
Lemma~\ref{lem:gen-col-ineq} we can use admissibility to obtain an
approximation of the $r$-colouring number of a given graph from a class of
bounded expansion.

Bounded expansion classes are very robust under local changes, e.g.\ under
taking lexicographic products, as defined below.

\begin{definition}
  Let $G$ and $H$ be graphs. The \emph{lexicographic product $G\bullet H$}
  of $G$ and $H$ is the graph with vertex set $V\coloneqq V(G)\times V(H)$
  and edge set
  \[E\coloneqq \bigl\{\{(v,v'),(u,u')\}\mid \text{$\{v,u\}\in E(G)$, or $v=u$ and $\{v',u'\}\in E(H)$}\bigr\}.\]
\end{definition}

The following lemma shows that taking lexicographic products preserves the
edge density of shallow minors. This was first proved
in~\cite{nevsetvril2008grad}; the following improved bounds are given
in~\cite{har2017approximation}.

\begin{lemma}[\rm Har-Peled and Quanrud
  \cite{har2017approximation}]\label{lem:stability-lex}
  For any graph $G$ and $r,t\in\N$ we have
  $\nabla_r(G\bullet K_t)\leq 5t^2(r+1)^2\nabla_r(G)$.
\end{lemma}

Bounded expansion classes are also stable under taking shallow minors, as
expressed in the following lemma.

\begin{lemma}[\rm see Ne\v{s}et\v{r}il and Ossona de Mendez~{\cite[Proposition~4.1]{nevsetvril2012sparsity}}]
\label{lem:stability-minors}
    If $J,H$ and $G$ are graphs and $r,s\in\N$ such that $J$ is a depth-$r$
    minor of $H$ and $H$ is a depth-$s$ minor of $G$, then $J$ is a
    depth-$(2rs+r+s)$-minor of $G$.
\end{lemma}

The following lemma is folklore, we provide a proof for completeness.

\begin{lemma}\label{lem:col-vs-nabla}
  For any graph $G$ and $r\in\N$ we have
    $\nabla_r(G)\leq \col_{4r+1}(G)$.
\end{lemma}

\begin{proof}
  Set $c=\col_{4r+1}(G)$ and let $L$ be a linear order of
    $V(G)$ for which $\col_{4r+1}(G,L)=c$. Next let $H\minor_r G$,
  say with a minor model $(I_u)_{u\in V(H)}$. We will show that
  $|E(H)|\leq c\cdot|V(H)|$.

  For each $u\in V(H)$ let $m_u$ be the $<_L$-minimal vertex in $I_u$. We
  define a linear order on $V(H)$ by setting $u<v$ if $m_u <_L m_v$.
  Observe that since each branch set has radius at most $r$ and $m_u$ and
  $m_v$ are minimum in their respective branch sets, if $u < v$, there
  exists a vertex in $I_u$ which is strongly $(4r+1)$-reachable from~$m_v$.
  Hence, $H$ is $c$-degenerate and can have at most $c\cdot|V(H)|$ edges.
\end{proof}

We are going to prove the following theorem. 

\begin{theorem}\label{thm:main}
  Let $\tau$ be a finite and purely relational signature and let $\CCC$ be
  a class of $\tau$-structures of bounded expansion. Then there exists an
  algorithm that, given a finite $\tau$-structure $\strA\in\CCC$ and a
  successor-invariant formula $\phi\in \FOs$, verifies whether
  $\strA\models_{\mathrm{succ-inv}}\phi$ in time
  $f(|\phi|)\cdot n\cdot \alpha(n)$, where~$f$ is a function and $n$ is the
  size of the universe of $\strA$.
\end{theorem}

In the language of parameterised complexity, Theorem~\ref{thm:main}
essentially states that the model-checking problem for successor-invariant
first-order formulas is fixed-parameter tractable on classes of finite
structures whose underlying Gaifman graph belongs to a fixed class of
bounded expansion. There is a minor caveat, though. The formal definition
of fixed-parameter tractability, see e.g.~\cite{FlumG06}, requires the
function $f$ to be computable, which is not asserted by
Theorem~\ref{thm:main}. In order to have this property, it suffices to
assume that the class $\CCC$ is \emph{effectively of bounded expansion}. In
the characterisation of Theorem~\ref{thm:b_exp_b_deg}, this means that
there exist a \emph{computable} function $f:\N\to\N$ such that
$\col_r(G(\strA))\le f(r)$ for each $\strA\in \CCC$. See
\cite{grohe2014deciding} for a similar discussion regarding model-checking
first-order logic on (effectively) nowhere dense classes of graphs.

\medskip
In principle, our approach follows that of the earlier results on
successor-invariant model-checking. As $\phi$ is successor-invariant, to
verify whether $\strA\models_{\mathrm{succ-inv}}\phi$, we may compute an arbitrary successor relation~$S$ on $V(\strA)$, and verify whether
$(\strA,S)\models \phi$. Of course, we will try to compute a successor
relation $S$ so that adding it to $\strA$ preserves the structural
properties as much as possible, so that model-checking on~$(\strA,S)$ can
be done efficiently.

Our construction of such a structure preserving successor relation is based
on the above described characterisation of bounded expansion classes
by the generalised colouring numbers.
As a first step, we show how to
define a set $F$ of new edges with the following properties:
\begin{itemize}
\item $F$ forms a tree on the vertex set $V(G)$ of
  the input graph $G$,
\item $F$ has maximum degree at most $3$, and 
\item after adding all the edges of $F$ to $G$, the
  colouring numbers are still bounded.
\end{itemize}
In a second step, we construct from the bounded degree spanning tree a
successor relation on~$V(G)$, again ensuring that the relevant
parameters remain bounded.

\subsection{Constructing a low-degree spanning tree}\label{sec:tree}

\newcommand{\Zup}{Z^{\uparrow}}
\newcommand{\Zdown}{Z^{\downarrow}}
\newcommand{\Fnew}{F_{\textrm{new}}}

In this section we prove the following theorem. 

\begin{theorem}\label{thm:main-technical}
  There exists an algorithm that, given a graph $G$, $r\in\N$, and
  ordering $L$ of $V(G)$, computes a set of unordered pairs
  $F\subseteq\binom{V(G)}{2}$ such that the graph $T=(V(G),F)$ is a tree of
  maximum degree at most $3$ and
  \[\adm_r(G+F,L)\le 2+2\cdot\col_{2r}(G,L),\]
  where $G+F$ is the graph $(V(G),E(G)\cup F)$. The running time of the
  algorithm is 
  $\Oh((n+m)\cdot\alpha(m))$, where $m=|E(G)|$ and $n=|V(G)|$.
\end{theorem}

The main step towards this goal is the corresponding statement for
connected graphs, as expressed in the following lemma.

\begin{lemma}\label{lem:connected}
  There exists an algorithm that, given a connected graph $G$,
    $r\in\N$, and ordering $L$ of $V(G)$, computes a set of
  unordered pairs $F\subseteq\binom{V(G)}{2}$ such that the graph
  $T=(V(G),F)$ is a tree of maximum degree at most $3$ and
  \[\adm_r(G+F,L)\le 2\cdot\col_{2r}(G,L).\]
  The running time of the algorithm is $\Oh((n+m)\cdot\alpha(m))$, where
  $m=|E(G)|$ and $n=|V(G)|$.
\end{lemma}

We first show that Theorem~\ref{thm:main-technical} follows easily from
Lemma~\ref{lem:connected}.

\begin{proof}[Proof of Theorem~\ref{thm:main-technical}, assuming
  Lemma~\ref{lem:connected}] \
  Let $G$ be a (possibly disconnected) graph, and let $G_1,\ldots,G_p$ be
  its connected components. For $i=1,\ldots,p$,
  let $L_i$ be the ordering obtained by restricting $L$ to $V(G_i)$.
  Obviously $\col_{2r}(G_i,L_i)\le \col_{2r}(G,L)$.

  Apply the algorithm of Lemma~\ref{lem:connected} to $G_i$ and $L_i$,
  obtaining a set of unordered pairs $F_i$ such that $T_i=(V(G_i),F_i)$ is
  a tree of maximum degree at most $3$ and
  \[\adm_r(G_i+F_i,L_i)\le 2\cdot\col_{2r}(G_i,L_i)\le
  2\cdot\col_{2r}(G,L).\]
  For $i=1,\ldots,p$, select a vertex $v_i$ of $G_i$ with degree at
  most $1$ in $T_i$; since $T_i$ is a tree, such a vertex exists. Define
  \[F= \{v_1v_2,v_2v_3,\ldots,v_{p-1}v_p\}\cup\bigcup_{i=1}^p F_i.\]
  Obviously we have that $T=(V(G),F)$ is a tree. Observe that it has
  maximum degree at most $3$. This is because each vertex $v_i$ had degree
  at most $1$ in its corresponding tree $T_i$, and hence its degree can
  grow to at most $3$ after adding edges $v_{i-1}v_i$ and $v_iv_{i+1}$. By
  Lemma~\ref{lem:connected}, the construction of each $T_i$ takes time
  $\Oh((n_i+m_i)\cdot\alpha(m_i))$, where $m_i=|E(G_i)|$. It follows that the construction of $T$ takes time $\Oh((n+m)\cdot\alpha(m))$.

  It remains to argue that $\adm_r(G+F,L)\le 2+2\col_{2r}(G,L)$. Take any
  vertex $u$ of $G$, say $u\in V(G_i)$, and let $\Pp$ be a set of paths of
  length at most $r$ that start in $u$, are pairwise vertex-disjoint (apart
  from $u$), and end in vertices smaller than $u$ in~$L$, while internally
  traversing only vertices larger than $u$ in~$L$. Observe that at most two
  of the paths from $\Pp$ can use any of the edges from the set
  $\{v_1v_2,v_2v_3,\ldots,v_{p-1}v_p\}$, since any such path has to use
  either $v_{i-1}v_i$ or $v_iv_{i+1}$. The remaining paths are entirely
  contained in $G_i+F_i$, and hence their number is bounded by
  $\adm_r(G_i+F_i,L_i)\le 2\col_{2r}(G,L)$. The theorem follows.
\end{proof}

In the remainder of this section we focus on Lemma~\ref{lem:connected}.

\begin{proof}[Proof of Lemma~\ref{lem:connected}] \
  We begin our proof by showing how to compute the set $F$. This will be a
  two step process, starting with an \emph{elimination tree}. For a
  connected graph~$G$ and an ordering~$L$ of $V(G)$, we define the
  \emph{(rooted) elimination tree} $S(G,L)$ of $G$ imposed by~$L$ (cf.~\cite{bodlaender1998rankings,schaffer1989optimal}) as follows. If
  $V(G)=\{v\}$, then the rooted elimination tree $S(G,L)$ is just the tree
  on the single vertex $v$. Otherwise, the root of $S(G,L)$ is the vertex
  $w$ that is the smallest with respect to the ordering~$L$ in~$G$. For
  each connected component $C$ of $G-w$ we construct a rooted elimination
  tree $S(C,L|_{V(C)})$, where $L|_{V(C)}$ denotes the restriction of $L$
  to the vertex set of $C$. These rooted elimination trees are attached
  below $w$ as subtrees by making their roots into children of $w$. Thus,
  the vertex set of the elimination tree $S(G,L)$ is always equal to the
  vertex set of $G$. See Figure~\ref{fig:G_S_U} for an illustration. The
  solid black lines are the edges of $G$; the dashed blue lines are the
  edges of $S$. The ordering $L$ is given by the numbers written in the
  vertices.

  \begin{figure*}
    \centering
    \begin{tikzpicture}[scale=0.8]
      \node[vertex] (0) at (0,0){$0$};
      \node[vertex] (1) at (0,-1){$1$};
      \node[vertex] (2) at (-3,-2){$2$};
      \node[vertex] (3) at (-6,-3){$3$};
      \node[vertex] (4) at (4,-2){$4$};
      \node[vertex] (5) at (7,-4.5){$5$};
      \node[vertex] (6) at (-1,-3){$6$};
      \node[vertex] (7) at (4,-3){$7$};
      \node[vertex] (8) at (-3,-3){$8$};
      \node[vertex] (9) at (6,-4.5){$9$};
      \node[vertex] (10) at (-6,-5){$10$};
      \node[vertex] (19) at (-1,-5){$19$};
      \node[vertex] (12) at (-4,-4){$12$};
      \node[vertex] (13) at (4,-4){$13$};
      \node[vertex] (14) at (-5,-4){$14$};
      \node[vertex] (15) at (3,-5){$15$};
      \node[vertex] (16) at (-2,-4){$16$};
      \node[vertex] (17) at (5.5,-6){$17$};
      \node[vertex] (18) at (5,-4.5){$18$};
      \node[vertex] (11) at (0,-4){$11$};
      \node[vertex] (20) at (4,-6){$20$};
      \node[vertex] (21) at (2,-6){$21$};
      \node[vertex] (22) at (3,-6){$22$};

      \draw[S] (0) -- (1);
      \draw[S] (1) -- (2);
      \draw[S] (1) -- (4);
      \draw[S] (2) -- (3);
      \draw[S] (2) -- (6);
      \draw[S] (4) -- (5);
      \draw[S] (14) -- (10);
      \draw[S] (3) -- (14);
      \draw[S] (2) -- (8);
      \draw[S] (8) -- (12);
      \draw[S] (8) -- (16);
      \draw[S] (11) -- (19);
      \draw[S] (6) -- (11);
      \draw[S] (4) -- (7);
      \draw[S] (7) -- (13);
      \draw[S] (7) -- (9);
      \draw[S] (13) -- (15);
      \draw[S] (13) -- (18);
      \draw[S] (15) -- (17);
      \draw[S] (15) -- (21);
      \draw[S] (15) -- (22);
      \draw[S] (15) -- (20);

      \draw[G] (0) to (1);
      \draw[G] (1) to (11);
      \draw[G] (6) to (11);
      \draw[G] (19) to (11);
      \draw[G] (2) to (6);
      \draw[G] (2) to (8);
      \draw[G] (8) to (12);
      \draw[G] (8) to (16);
      \draw[G] (2) to (14);
      \draw[G] (3) to (14);
      \draw[G] (3) to (10);
      \draw[G] (1) to (21);
      \draw[G] (21) to (15);
      \draw[G] (15) to (22);
      \draw[G] (15) to (20);
      \draw[G,bend left=7] (15) to (17);
      \draw[G] (15) to (13);
      \draw[G] (15) to (7);
      \draw[G] (15) to (4);
      \draw[G] (4) to (5);
      \draw[G] (7) to (9);
      \draw[G] (13) to (18);
      \draw[G] (14) to (10);
      \draw[G,bend left=20] (0) to (21);

      \draw[U] (0) to (1);
      \draw[U] (1) to (11);
      \draw[U] (6) to (11);
      \draw[U] (19) to (11);
      \draw[U] (2) to (6);
      \draw[U] (2) to (8);
      \draw[U] (8) to (12);
      \draw[U] (8) to (16);
      \draw[U] (2) to (14);
      \draw[U] (3) to (14);
      \draw[U] (14) to (10);
      \draw[U,bend left=15] (1) to (21);
      \draw[U] (21) to (15);
      \draw[U] (15) to (22);
      \draw[U] (15) to (20);
      \draw[U,bend right=7] (15) to (17);
      \draw[U,bend left=15] (15) to (4);
      \draw[U,bend left=15] (15) to (7);
      \draw[U] (15) to (13);
      \draw[U] (13) to (18);
      \draw[U] (7) to (9);
      \draw[U] (4) to (5);
    \end{tikzpicture}
    \caption[A graph $G$, the elimination tree $S$ and a tree $U$]%
    {A graph $G$ (solid black lines), the elimination tree $S$ (dashed blue
      lines), and a tree $U$ (dotted red lines).}
  \label{fig:G_S_U}
\end{figure*}
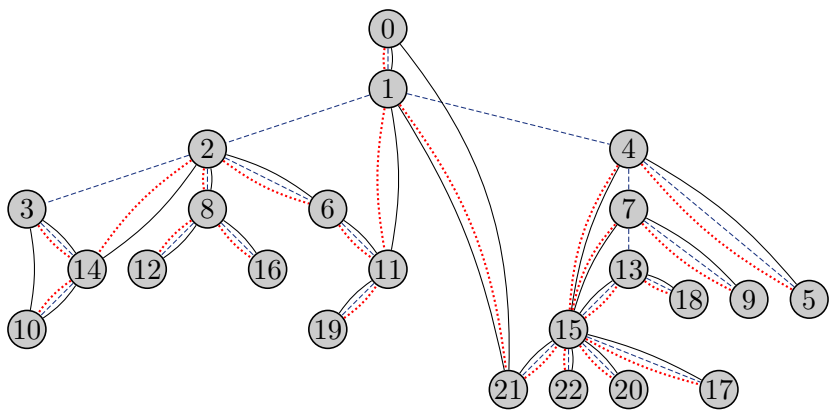
  
Let $S=S(G,L)$ be the rooted elimination tree of $G$ imposed by $L$. For a
vertex $u$, by $G_u$ we denote the subgraph of $G$ induced by all
descendants of $u$ in $S$, including $u$. The following properties follow
easily from the construction of a rooted elimination tree.

\begin{claim}\label{prop:eltree}
  The following assertions hold.
  \begin{enumerate}
  \item For each $u\in V(G)$, the subgraph $G_u$ is
    connected.\label{item:connected}
  \item Whenever a vertex $u$ is an ancestor of a vertex $v$ in $S$, we
    have $u\le_Lv$.\label{item:SL}
  \item For each $uv\in E(G)$ with $u<_Lv$, $u$ is an ancestor of~$v$ in
    $S$.\label{item:acestor}
  \item For each $u\in V(G)$ and each child $v$ of $u$ in $S$, $u$ has at
    least one neighbour in $V(G_v)$.\label{item:exists_child}
  \end{enumerate}
\end{claim}

\begin{proof}
  Assertions~(\ref{item:connected}) and~(\ref{item:SL}) follow immediately
  from the construction of $S$. For assertion~(\ref{item:acestor}), suppose
  that $u$ and $v$ are not bound by the ancestor-descendant relation in
  $S$, and let $w$ be their lowest common ancestor in $S$. Then $u$ and $v$
  would be in different connected components of $G_w-w$, hence $uv$ could
  not be an edge; a contradiction. It follows that $u$ and $v$ are bound by
  the ancestor-descendant relation, implying that~$u$ is an ancestor of
  $v$, due to $u<_L v$ and assertion~(\ref{item:SL}). Finally, for
  assertion~(\ref{item:exists_child}), recall that by
  assertion~(\ref{item:connected}) we have that $G_u$ is connected, whereas
  by construction $G_v$ is one of the connected components of $G_u-u$.
  Hence, in $G$ there is no edge between $V(G_v)$ and any of the other
  connected components of $G_u-u$. If there was no edge between $V(G_v)$
  and $u$ as well, then there would be no edge between $V(G_v)$ and
  $V(G_u)\setminus V(G_v)$, contradicting the connectivity of~$G_u$.
\end{proof}

We now define a set of edges $B\subseteq E(G)$ as follows. For every vertex
$u$ of $G$ and every child $v$ of $u$ in~$S$, select an arbitrary neighbour
$w_{u,v}$ of $u$ in $G_v$; such a neighbour exists by
Claim~\ref{prop:eltree}\,(\ref{item:exists_child}). Then let~$B_u$ be the
set of all edges $uw_{u,v}$, for~$v$ ranging over the children of $u$
in~$S$. Define
\[B=\bigcup_{u\in V(G)} B_u.\]
Let $U$ be the graph spanned by all the edges in $B$, that is,
$U=(V(G),B)$. In Figure~\ref{fig:G_S_U}, the edges of $U$ are represented
by the dotted red lines.

\begin{claim}\label{lem:is-tree}
  The graph $U$ is a tree.
\end{claim}

\begin{proof}
  Observe that for each $u\in V(G)$, the number of edges in $B_u$ is equal
  to the number of children of $u$ in $S$. Since every vertex of $G$ has
  exactly one parent in $S$, apart from the root of~$S$, we infer that
  \[|B|\le \sum_{u\in V(G)}|B_u|= |V(G)|-1.\]
  Therefore, since $B$ is the edge set of $U$, to prove that~$U$ is a tree
  it suffices to prove that $U$ is connected. To this end, we prove by a
  bottom-up induction on $S$ that for each $u\in V(G)$, the subgraph
  $U_u=\bigl(V(G_u),B\cap \binom{V(G_u)}{2}\bigr)$ is connected. Note that
  for the root~$w$ of~$S$ this claim is equivalent to $U_w=U$ being
  connected.

  Take any $u\in V(G)$, and suppose by induction that for each child $v$
  of~$u$ in $S$, the subgraph~$U_v$ is connected. Observe that~$U_u$ can be
  constructed by taking the vertex $u$ and, for each child $v$ of~$u$
  in~$S$, adding the connected subgraph~$U_v$ and connecting~$U_v$
  to $u$ via the edge $uw_{u,v}\in B_u$. Thus, $U_u$ constructed in
  this manner is also connected, as claimed.
\end{proof}

Next, we verify that $U$ can be computed within the claimed running time.
Note that we do not need to compute $S$, as we use it only in the analysis.
We remark that this is the only place in the algorithm where the running
time is not linear.

\begin{claim}
  The tree $U$ can be computed in time $\Oh(m\cdot\alpha(m))$.
\end{claim}

\begin{proof}
  We use the classic Union~\&~Find data structure on the set $V(G)$.
  Recall that in this data structure, at each moment we maintain a
  partition of $V(G)$ into a number of equivalence classes, each with a
  prescribed representative, where initially each vertex is in its own
  class. The operations are a) for a given $u\in V(G)$, find the
  representative of the class to which $u$ belongs, and b) merge two
  equivalence classes into one. Tarjan~\cite{Tarjan75} gave an
  implementation of this data structure where both operations run in
  amortised time $\alpha(k)$, where $k$ is the total number of operations
  performed.

  Having initialised the data structure, we process the vertex ordering $L$
  from the smallest end, starting with an empty prefix. For an
  already processed prefix $X$ of $L$, the maintained classes
  within~$X$ will represent the partition of $G[X]$ into connected
  components, while every vertex outside~$X$ will still be in its own
  equivalence class. Let us consider one step, when we process a vertex
  $u$, thus moving from a prefix $X$ to the prefix
  $X'=X\cup \{u\}$. Iterate through all the neighbours of~$u$, and for each
  neighbour $v$ of $u$ such that $u<_L v$, verify whether the equivalence
  classes of $u$ and~$v$ are different. If this is the case, merge these
  classes and add the edge $uv$ to $B$. A straightforward induction shows
  that the claimed invariant holds. Moreover, when processing $u$ we add
  exactly the edges of $B_u$ to $B$, hence at the end we obtain the set $B$
  and the tree $U=(V(G),B)$.

  For the running time analysis, observe that in total we perform $\Oh(m)$
  operations on the data structure, thus the running time is
  $\Oh(m\alpha(m))$. We remark that we assume that the ordering $L$ is
  given as a bijection between $V(G)$ and numbers $\{1,2,\ldots,|V(G)|\}$,
  thus for two vertices $u,v$ we can check in constant time whether
  $u<_L v$.
\end{proof}

By Lemma~\ref{lem:is-tree} we have that $U$ is a spanning tree
of~$G$. However its maximum degree may be too large. The idea is to
use $U$ to construct a new tree~$T$ with maximum degree at most~$3$ (on the
same vertex set $V(G)$). The way we constructed~$U$ will enable us to argue
that adding the edges of~$T$ to the graph $G$ does not change the
generalised colouring numbers too much.

Give $U$ the same root as the elimination tree $S$. From now on we treat
$U$ as a rooted tree, which imposes parent-child and ancestor-descendant
relations in $U$ as well. Note that the parent-child and
ancestor-descendant relations in~$S$ and in $U$ may be completely
different. For instance, consider vertices $4$ and $15$ in the example from
Figure~\ref{fig:G_S_U}: $4$ is a child of $15$ in $U$, and an ancestor of
$15$ in $S$.

\begin{figure*}
  \centering
  \begin{tikzpicture}[scale=0.8]
    \node[vertex] (0) at (0,0){$0$};
    \node[vertex] (1) at (0,-1){$1$};
    \node[vertex] (2) at (-3,-2){$2$};
    \node[vertex] (3) at (-6,-3){$3$};
    \node[vertex] (4) at (4,-2){$4$};
    \node[vertex] (5) at (7,-4.5){$5$};
    \node[vertex] (6) at (-1,-3){$6$};
    \node[vertex] (7) at (4,-3){$7$};
    \node[vertex] (8) at (-3,-3){$8$};
    \node[vertex] (9) at (6,-4.5){$9$};
    \node[vertex] (10) at (-6,-5){$10$};
    \node[vertex] (19) at (-1,-5){$19$};
    \node[vertex] (12) at (-4,-4){$12$};
    \node[vertex] (13) at (4,-4){$13$};
    \node[vertex] (14) at (-5,-4){$14$};
    \node[vertex] (15) at (3,-5){$15$};
    \node[vertex] (16) at (-2,-4){$16$};
    \node[vertex] (17) at (5.5,-6){$17$};
    \node[vertex] (18) at (5,-4.5){$18$};
    \node[vertex] (11) at (0,-4){$11$};
    \node[vertex] (20) at (4,-6){$20$};
    \node[vertex] (21) at (2,-6){$21$};
    \node[vertex] (22) at (3,-6){$22$};

    \draw[G] (0) to (1);
    \draw[G] (1) to (11);
    \draw[G] (6) to (11);
    \draw[G] (19) to (11);
    \draw[G] (2) to (6);
    \draw[G] (2) to (8);
    \draw[G] (8) to (12);
    \draw[G] (8) to (16);
    \draw[G] (2) to (14);
    \draw[G] (3) to (14);
    \draw[G] (3) to (10);
    \draw[G] (1) to (21);
    \draw[G] (21) to (15);
    \draw[G] (15) to (22);
    \draw[G] (15) to (20);
    \draw[G,bend left=7] (15) to (17);
    \draw[G] (15) to (13);
    \draw[G] (15) to (7);
    \draw[G,bend left=30] (15) to (4);
    \draw[G] (4) to (5);
    \draw[G] (7) to (9);
    \draw[G] (13) to (18);
    \draw[G] (14) to (10);
    \draw[G,bend left=20] (0) to (21);

    \draw[U] (0) to (1);
    \draw[U] (1) to (11);
    \draw[U] (6) to (11);
    \draw[U] (19) to (11);
    \draw[U] (2) to (6);
    \draw[U] (2) to (8);
    \draw[U] (8) to (12);
    \draw[U] (8) to (16);
    \draw[U] (2) to (14);
    \draw[U] (3) to (14);
    \draw[U] (14) to (10);
    \draw[U,bend left=15] (1) to (21);
    \draw[U] (21) to (15);
    \draw[U] (15) to (22);
    \draw[U] (15) to (20);
    \draw[U,bend right=7] (15) to (17);
    \draw[U,bend left=15] (15) to (4);
    \draw[U,bend left=15] (15) to (7);
    \draw[U] (15) to (13);
    \draw[U] (13) to (18);
    \draw[U] (7) to (9);
    \draw[U] (4) to (5);

    \draw[T] (0) to (1);
    \draw[T] (1) to (11);
    \draw[T] (6) to (11);
    \draw[T] (6) to (19);
    \draw[T] (2) to (6);
    \draw[T] (14) to (8);
    \draw[T] (8) to (12);
    \draw[T] (12) to (16);
    \draw[T] (2) to (8);
    \draw[T] (3) to (14);
    \draw[T] (3) to (10);
    \draw[T] (11) to (21);
    \draw[T] (21) to (15);
    \draw[T] (13) to (17);
    \draw[T] (22) to (20);
    \draw[T] (20) to (17);
    \draw[T,bend left=23] (15) to (4);
    \draw[T] (4) to (7);
    \draw[T] (7) to (13);
    \draw[T] (13) to (18);
    \draw[T] (7) to (9);
    \draw[T] (4) to (5);
  \end{tikzpicture}
  \caption[The graph $G$, the tree $U$ and a tree $T$]%
  {A graph $G$ (solid black lines), a tree $U$ (dotted red lines), and the
    tree $T$ (thick dashed green lines).}
  \label{fig:G_U_T}
\end{figure*}
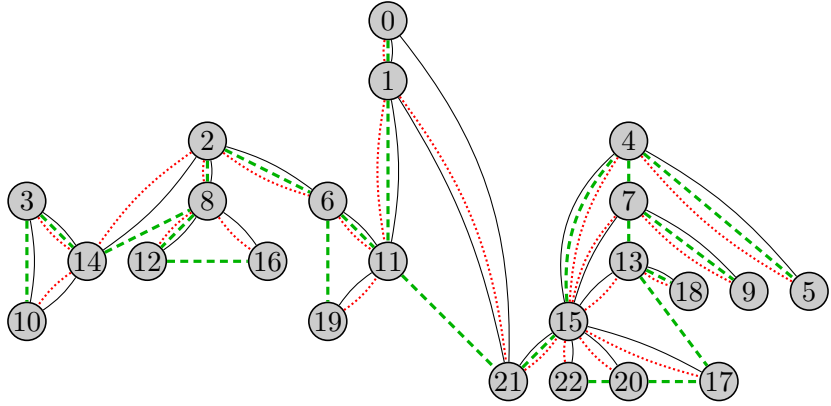

For every $u\in V(G)$, let $(x_1,\ldots,x_p)$ be an enumeration of the
children of $u$ in $U$, such that $x_i<_L x_j$ if $i<j$. Let
$F_u=\{ux_1,x_1x_2,x_2x_3,\ldots,x_{p-1}x_{p}\}$, and define
\[F= \bigcup_{u\in V(G)}F_u\quad \text{and}\quad T=(V(G),F).\]
See Figure~\ref{fig:G_U_T} for an illustration.

\begin{claim}\label{lem:correctness}
  The graph $T$ is a tree with maximum degree at most~$3$.
\end{claim}

\begin{proof}
  Observe that for each $u\in V(G)$ we have that $|F_u|$ is equal to the
  number of children of~$u$ in $U$. Every vertex of $G$ apart from the root
  of $U$ has exactly one parent in $U$, hence
  \[|F|\le \sum_{u\in V(G)}|F_u|=|V(G)|-1.\]
  Therefore, to prove that $T$ is a tree, it suffices to argue that it is
  connected. This, however, follows immediately from the fact that $U$ is
  connected, since for each edge in $U$ there is a path in $T$ that
  connects the same pair of vertices.

  Finally, it is easy to see that each vertex $u$ is incident to at most
  $3$ edges of $F$: at most one leading to a child of $u$ in $U$, and at
  most $2$ belonging to $F_v$, where $v$ is the parent of $u$ in $U$.
\end{proof}

Observe that once the tree $U$ is constructed, it is straightforward to
construct $T$ in time $\Oh(n)$. Thus, it remains to check that adding $F$
to $G$ does not change the generalised colouring numbers too much.

Take any vertex $u\in V(G)$ and examine its children in~$U$. We partition
them as follows. Let~$\Zup_u$ be the set of those children of~$u$ in $U$
that are its ancestors in $S$, and let $\Zdown_u$ be the set of those
children of $u$ in $U$ that are its descendants in~$S$. By the construction
of $U$ and by Claim~\ref{prop:eltree}.\ref{item:acestor}, each child of~$u$
in $U$ is either its ancestor or descendant in $S$. By
Claim~\ref{prop:eltree}.\ref{item:SL}, this is equivalent to saying that
$\Zup_u$, respectively $\Zdown_u$, comprise the children of $u$ in $U$ that
are smaller, respectively larger, than $u$ in $L$. Note that by the
construction of $U$, the vertices of $\Zdown_u$ lie in pairwise different
subtrees rooted at the children of $u$ in $S$, thus $u$ is the lowest
common ancestor in $S$ of every pair of vertices from $\Zdown_u$. On the
other hand, all vertices of $\Zup_u$ are ancestors of $u$ in~$S$, thus
every pair of them is bound by the ancestor-descendant relation in $S$.

\begin{claim}\label{lem:adm-bound}
  The graph union $G+F$ satisfies $\adm_r(G+F,L)\le 2\cdot\col_{2r}(G,L)$.
\end{claim}

\begin{proof}
  Write $H=G+F$. Let $\Fnew=F\setminus E(G)$ be the set of edges from $F$
  that were not already present in $G$. If an edge $e\in\Fnew$ belongs also
  to $F_u$ for some $u\in V(G)$, then we know that $u$ cannot be an
  endpoint of $e$. This is because edges joining a vertex $u$ with its
  children in~$U$ were already present in $G$. We say that the vertex~$u$
  is the \emph{origin} of an edge $e\in\Fnew\cap F_u$, and denote it by
  $a(e)$. Observe that $a(e)$ is adjacent to both endpoints of $e$ in~$G$
  by construction. Also observe that if the endpoints of~$e$ belong to
  $\Zup_{a(e)}$, then they are both ancestors of $a(e)$ in $S$, and thus
  are both smaller than~$a(e)$ in $L$. Otherwise, if the endpoints of~$e$
  belong to~$\Zdown_{a(e)}$, then they are not bound by the
  ancestor-descendant relation in $S$ and $a(e)$ is their lowest common
  ancestor in $S$.

  To give an upper bound on $\adm_r(H,L)$, let us fix a vertex $u\in V(G)$
  and a family of paths~$\Pp$ in $H$ such that
  \begin{itemize}
  \item each path in $\Pp$ has length at most $r$, starts in $u$, ends in a
    vertex smaller than $u$ in $L$, and all its internal vertices are
    larger than $u$ in $L$;
  \item the paths in $\Pp$ are pairwise vertex-disjoint, apart from the
    starting vertex $u$.
  \end{itemize}

  For each path $P\in \Pp$, we define a walk $P'$ in $G$ as follows. For
  every edge $e=xy$ from~$\Fnew$ traversed on $P$, replace the usage of
  this edge on $P$ by the following detour of length $2$: $x{-}a(e){-}y$.
  Observe that $P'$ is a walk in the graph $G$, it starts in $u$, ends in
  the same vertex as $P$, and has length at most~$2r$. Next, we define
  $v(P)$ to be the first vertex on $P'$ (that is, the closest to~$u$ on
  $P'$) that does not belong to $G_u$. Since the endpoint of $P'$ that is
  not $u$ does not belong to~$G_u$, such a vertex exists. Finally, let
  $P''$ be the prefix of $P'$ from $u$ to the first visit of $v(P)$ on $P'$
  (from the side of $u$). Observe that the predecessor of $v(P)$ on $P''$
  belongs to $G_u$ and is a neighbour of $v(P)$ in $G$, hence $v(P)$ has to
  be a strict ancestor of $u$ in~$S$. We find that $P''$ is a walk of
  length at most $2r$ in $G$, it starts in $u$, ends in $v(P)$, and all its
  internal vertices belong to $G_u$, so in particular they are not smaller
  than $u$ in $L$. This means that~$P''$ certifies that
  $v(P)\in\SReach_{2r}[G,L,u]$.

  Since $\bigl|\SReach_{2r}[G,L,u]\bigr|\le \col_{2r}(G,L)$, in order to
  prove the bound on $\adm_r(H,L)$, it suffices to prove the following
  claim: For each vertex $v$ that is a strict ancestor of $u$ in~$S$, there
  can be at most two paths $P\in\Pp$ for which $v=v(P)$. To this end, we
  fix a vertex $v$ that is a strict ancestor of~$u$ in~$S$ and proceed by a
  case distinction on how a path $P$ with $v=v(P)$ may behave.

  Suppose first that $v$ is the endpoint of $P$ other than $u$,
  equivalently the endpoint of $P'$ other than~$u$. (For example, $u=1$,
  $P= 1,11,21,0$, $P'=1,11,1,21,0$ and $v=0$, in
  figures~\ref{fig:G_S_U} and~\ref{fig:G_U_T}.) However, the paths of $\Pp$
  are pairwise vertex-disjoint, apart from the starting vertex $u$, hence
  there can be at most one path $P$ from $\Pp$ for which $v$ is an
  endpoint. Thus, this case contributes at most one path $P$ for which
  $v=v(P)$.

  Next suppose that $v$ is an internal vertex of the walk $P'$; in
  particular, it is not the endpoint of~$P$ other than $u$. (For example,
  $u=6$, $P= 6,11,21,0$, $P'= 6,11,1,21,0$ and $v=1$, in
  figures~\ref{fig:G_S_U} and~\ref{fig:G_U_T}.) Since the only vertex
  traversed by~$P$ that is smaller than $u$ in $L$ is this other endpoint
  of~$P$, and $v$ is smaller than $u$ in $L$ due to being its strict
  ancestor in $S$, it follows that each visit of $v$ on $P'$ is due to
  having $v=a(e)$ for some edge $e\in\Fnew$ traversed on~$P$. Select $e$ to
  be such an edge corresponding to the first visit of $v$ on $P'$. Let
  $e=xy$, where $x$ lies closer to $u$ on $P$ than $y$. (That is, in our
  figures, $x=11$ and $y=21$.) Since $v$ was chosen as the first
  vertex on $P'$ that does not belong to $G_u$, we have $x\in G_u$.

  Since $v=a(e)=a(xy)$, either $x\in\Zdown_v$ or $x\in\Zup_v$. Note that
  the second possibility cannot happen, because then~$v$ would be a
  descendant of $x$ in $S$, hence $v$ would belong to $G_u$, due to
  $x\in G_u$; a contradiction. We infer that $x\in\Zdown_v$.

  Recall that, by construction, $\Zdown_v$ contains at most one vertex from
  each subtree of $S$ rooted at a child of $v$. Since $v$ is a strict
  ancestor of $u$ in $S$, we infer that $x$ has to be the unique vertex of
  $\Zdown_v$ that belongs to $G_u$. In the construction of $F_v$, however,
  we added only at most two edges of $F_v$ incident to this unique vertex:
  at most one to its predecessor on the enumeration of the children of $v$,
  and at most one to its successor. Since paths from $\Pp$ are pairwise
  vertex-disjoint in~$H$, apart from the starting vertex $u$, only at most
  one path from $\Pp$ can use these two edges. We can have $v=a(e)$ for
  this path only. Thus, this case contributes at most one path $P$ for
  which $v=v(P)$, completing the proof of the claim.
\end{proof}

We conclude the proof by summarising the algorithm: first construct the
tree $U$, and then construct the tree $T$. As argued, these steps take time
$\Oh(m\cdot\alpha(m))$ and $\Oh(n)$, respectively. By
claims~\ref{lem:correctness} and~\ref{lem:adm-bound}, $T$ satisfies the
required properties.
\end{proof}

\subsection{Constructing a successor relation}

The preceding section provides us with a spanning tree of maximum degree at
most $3$. We now show how this can be used to obtain a successor relation
from this spanning tree.

We give two constructions: One which constructs an actual successor
relation, at the cost of possibly adding further edges. The added edges may
increase the admissibility, but in a way that preserves bounded expansion. We also give a second construction that does not add additional edges
  and hence preserves also other structural properties. Such a construction
  may thus be potentially used for model-checking on other graph classes. 
This construction shows how a successor relation may be first-order
interpreted in a graph with bounded-degree spanning tree, without adding
any edges.

\pagebreak
\paragraph*{Adding a successor relation.}
\begin{theorem}\label{thm:tree-to-succ}
  There exists an algorithm that, given a graph $G$, $r\in\N$, the
  edge set $F$ of a tree of maximum degree at most $3$, and an ordering $L$
  of $V(G)$, computes a set of ordered pairs $S\subseteq V(G)^2$
  such that~$S$ is a successor relation on $V(G)$ and
  \[\adm_r(G+\bar S,L)\le h\bigl(r,\adm_{40r+1}(G+F)\bigr),\]
  for an appropriately defined function $h$ where $\bar S = \left\{\{a,b\}
  \mid (a,b) \in S, a\neq b\right\}$. 
  The running time of the algorithm is $\Oh(m+n)$, where
  $m=|E(G)|$ and $n=|V(G)|$.
\end{theorem}

As observed e.g.\ in~\cite{karaganis1968cube,sekanina1963ordering}, the
cube of every connected graph contains a Hamiltonian path. (The
  \emph{cube} of a graph $G$ is the graph on the same vertex set as $G$ and
  in which two vertices are connected if their distance in $G$ is at most
  $3$.) Furthermore, such a Hamiltonian path can be computed in linear
time in the size of the original graph~\cite{lin1995algorithms}. The set
$S$ of edges whose existence is stated in Theorem~\ref{thm:tree-to-succ}
will simply be the Hamiltonian path computed in the cube of the spanning
tree $F$ that we constructed above. It remains to prove the claimed
  bound on the $r$-admissibility of the new graph.

\begin{proof}[Proof of Theorem~\ref{thm:tree-to-succ}] \
  Observe that we can find $G+S$, where $S$ is as described above, as a
  depth-$3$ minor of $(G+F)\bullet K_9$. This is a simple consequence of
  the fact that $F$ has maximum degree~$3$. Now we have
  \begin{align*}
    \nabla_r(G+S)&\leq \nabla_r\bigl(\nabla_3((G+F)\bullet K_9)\bigr)\\
    &\leq \nabla_{10r}\bigl((G+F)\bullet K_9\bigr)
    &&\text{(by Lemma~\ref{lem:stability-minors})}\\
    &\leq 5\cdot 9^2\cdot (10r+1)^2\cdot \nabla_{10r}(G+F)
    &&\text{(by Lemma~\ref{lem:stability-lex})}\\
    &\leq 5\cdot 9^2\cdot (10r+1)^2\cdot \col_{40r+1}(G+F)
    &&\text{(by Lemma~\ref{lem:col-vs-nabla}).}
  \end{align*}
  Finally, by Lemma~\ref{lem:adm-nabla} we have
  $\adm_r(G+S)\leq 6r(\nabla_r(G+S))^3$, which gives us
  $\adm_r(G+S)\leq g\bigl(r,\col_{40r+1}(G+F)\bigr)$,
  for an appropriately defined function $g$. We have $\col_{40r+1}(G+F)
  \leq \adm_{40r+1}(G+F)^{40r+1}$ by
  Lemma~\ref{lem:gen-col-ineq}, which leads to the stated
  result. 
\end{proof}

\paragraph*{Interpreting a successor relation.}
We show how in a graph with a spanning tree of degree~$3$, a successor
relation can be interpreted after suitably colouring vertices and edges,
but without adding further edges. We first notice that existence of such a
spanning tree guarantees the existence of a $3$-walk, i.e.\ a walk through
the graph that visits each vertex at least once and at most three times.
The following lemma allows us to interpret a successor relation from a
$k$-walk in first-order logic, for arbitrary $k$. For a natural number
$\ell$, let $[\ell]$ be the set $\{1,\ldots,\ell\}$.

\begin{lemma}\label{lem:succfromkwalk}
  Let $\sigma$ be a finite relational signature, $\strA$ a finite
  $\sigma$-structure, and $w:[n]\to V(\strA)$ a $k$-walk through the
  Gaifman graph of $\strA$, where $n\coloneqq|V(\strA)|$.
  Then there is a finite relational signature $\sigma_k$ and a first-order
  formula $\varphi^{(k)}_{\mathrm{succ}}(x,y)$, both depending only on $k$,
  and a $(\sigma\cup\sigma_k)$-expansion $\strA'$ of $\strA$ which can be
  computed from $\strA$ and $w$ in polynomial time, such that
  \begin{itemize}
  \item the Gaifman-graphs of $\strA'$ and $\strA$ are the same;
  \item $\varphi^{(k)}_{\mathrm{succ}}$ defines a successor relation on
    $\strA'$.
  \end{itemize}
\end{lemma}

\begin{proof}
  We define a function $f:[n]\to [k]$ which counts how many times we have
  visited a vertex on the walk before, by
  \[f(i)\coloneqq |\{ j\leq i\mid w(i)=w(j) \}|.\]
  Furthermore, let $F:V(\strA)\to [k]$ count how many times we visit a
  vertex:
  \[F(v)\coloneqq |\{ i\in[n]\mid w(i)=v \}|.\]
  To simplify notation, if $i\in[n]$ we write $F(i)$ for $F(w(i))$.

  We encode the $k$-walk $w$ by binary relations $E_{ab}$ with
  $a,b=1,\ldots,k$, in such a way that $(u,v) \in E_{ab}$ if and only
    if there is some $i\in[n-1]$ such that
  \begin{itemize}
  \item $w(i)=u$ and $f(i)=a$, and
  \item $w(i+1)=v$ and $f(i+1)=b$.
  \end{itemize}
  That is, after visiting $u$ for the $a$-th time, the walk $w$ proceeds to
  $v$, visiting it for the $b$-th time. Note that if $k=1$, we can
  immediately define a successor relation by
  \[\varphi^{(1)}_{\mathrm{succ}}(x,y)\coloneqq E_{11}xy.\]

  If $k>1$, we show how to interpret a $(k-1)$-walk $w'$ in first-order
  logic, given a $k$-walk encoded by $\{E_{ab} \mid 1\leq a,b\leq k\}$ as
  above. By daisy-chaining these interpretations we end up with a $1$-walk
  (i.e.\ a Hamiltonian path). Plugging in the interpretation of this
  Hamiltonian path into $\varphi^{(1)}_{\mathrm{succ}}$ defined above
  gives the formulas $\varphi^{(k)}_{\mathrm{succ}}$.

  In order to get from a $k$-walk to a $(k-1)$-walk, we look at all
  vertices that are visited $k$ times, and ``jump'' over these vertices,
  either when they are visited for the $(k-1)$-th or for the $k$-th time.
  Jumping over a vertex can be done in first-order logic, but we must be
  careful to choose the vertices for jumping in such a way that we never
  jump over an unbounded number of vertices in a row, as this is not
  possible in first-order logic. We encode the information on whether to
  jump when visiting for the $(k-1)$-th or the $k$-th time in a new unary
  predicate $P_k$.

  To be precise, let $\varphi_{\text{$k$-times}}(x)$ be a formula which
  states that $x$ is visited $k$ times:
  \[\varphi_{\text{$k$-times}}(x)\coloneqq
  \bigvee_{a=1}^k\exists y\, E_{ka}xy.\]

  For those $u \in V(\strA)$ which are visited $k$-times, we agree to jump
  over them when they are visited for the $k$-th time if $u \in P_k$, and
  when they are visited for the $(k-1)$-th time otherwise. Thus, if $w(i) =
  u$, $f(i) = k$ and $u \in P_k$, we want to remove the $i$-th step.
  However, it may be the case that $w(i+1)$ is also visited $k$ times and
  needs to be jumped over. We define first-order formulas which carry out a
  bounded number of such jumps as follows.

  \begin{itemize}
  \item For $a\in [k]$, the formula $\varphi_{\text{jump},a}(x)$ holds if
    we jump over $x$ when visiting it for the $a$-th time:
    \begin{align*}
      \varphi_{\text{jump},1}(x),\ldots,\varphi_{\text{jump},k-2}(x)
      &\coloneqq \bot,\\
      \varphi_{\text{jump},k-1}(x)
      &\coloneqq \varphi_{\text{$k$-times}}(x)\wedge \neg P_kx,\\
      \varphi_{\text{jump},k}(x)
      &\coloneqq \varphi_{\text{$k$-times}}(x)\wedge P_kx.
    \end{align*}
  \item For $r\geq0$ and $a,b\in[k]$, the formula
    $\varphi_{\text{next},a,b}^{(r)}(x,y)$ holds if, when applying at most
    $r$ consecutive jumps on entering $x$ for the $a$-th time, we end up in
    node $y$ which is visited for the $b$-th time in the (original) walk.
    Specifically:
    \begin{align*}
      \varphi_{\text{next},a,b}^{(0)}(x,y)\coloneqq{}&
      x \dot= y\wedge \delta_{ab},\\
      \varphi_{\text{next},a,b}^{(r+1)}(x,y)\coloneqq{}&
      \bigl(\neg\varphi_{\text{jump},a}(x)\to
      (x\dot=y\wedge \delta_{ab})\bigr)\\
      &\wedge \Bigl(\varphi_{\text{jump},a}(x)\to
      \exists z\,\bigvee_{c=1}^k
      \bigl(E_{ac}xz\wedge \varphi_{\text{next},c,b}^{(r)}(z,y)\bigr)\Bigr).
    \end{align*}
    Here, $\delta_{ab}$ is true if the indices $a$ and $b$ are the same:
    \[\delta_{ab} \coloneqq \begin{cases}
      \top,&\text{if $a=b$};\\
      \bot,&\text{otherwise}.
    \end{cases}\]
  \item We will show below how to choose the predicate $P_k$ so that we
    never need to take more than two consecutive jumps. Thus, we can
    interpret a $(k-1)$-walk $w'$ using, for $a,b\in [k-2]$, the formulas
    \[\varphi_{E,a,b}(x,y)\coloneqq
    \exists z\, \bigvee_{c=1}^k
    \bigl(E_{ac}xz\wedge \varphi^{(2)}_{\text{next},c,b}(z,y)\bigr).\]
    For $a \in [k-2]$ we set
    \[\varphi_{E,a,k-1}(x,y)\coloneqq
    \exists z\, \bigvee_{c=1}^k
    \Bigl(E_{ac}xz\wedge
    \bigl(\varphi^{(2)}_{\text{next},c,k-1}(z,y)
    \vee\varphi^{(2)}_{\text{next},c,k}(z,y)\bigr)\Bigr).\]
    Next, for $b\in [k-2]$ we set
    \begin{align*}
      \varphi_{E,k-1,b}(x,y)\coloneqq{}&
      \Bigl(\neg\varphi_{\text{jump},k-1}(x)\to
      \exists z\,\bigvee_{c=1}^k
      \bigl(E_{k-1,c}xz\wedge \varphi^{(2)}_{\text{next},c,b}(z,y)\big)\Big)\\
      &\wedge \Bigl(\varphi_{\text{jump},k-1}(x)\to
      \exists z\,\bigvee_{c=1}^k \bigl(E_{k,c}xz\wedge
      \varphi^{(2)}_{\text{next},c,b}(z,y)\bigr)\Bigr),
    \end{align*}
    and finally we define
    \begin{align*}
      \varphi_{E,k-1,k-1}(x,y)\coloneqq{}&
      \biggl(\neg\varphi_{\text{jump},k-1}(x)\\
      &\qquad\to \exists z\,\bigvee_{c=1}^k
      \Bigl(E_{k-1,c}xz\wedge
      \bigl(\varphi^{(2)}_{\text{next},c,k-1}(z,y)\vee
      \varphi^{(2)}_{\text{next},c,k}(z,y)\bigr)\Bigr)\biggr)\\
      &\wedge \biggl(\varphi_{\text{jump},k-1}(x)\\
      &\qquad\to \exists z\,\bigvee_{c=1}^k
      \Bigl(E_{k,c}xz\wedge
      \bigl(\varphi^{(2)}_{\text{next},c,k-1}(z,y)\vee
      \varphi^{(2)}_{\text{next},c,k}(z,y)\bigr)\Bigr)\biggr).
    \end{align*}
  \end{itemize}
  

  To define the predicate $P_k$, let $T\subseteq [n]$ be the set of indices
  $i\in [n]$ for which $F(i)=k$ and $f(i)\in\{k-1,k\}$. We obtain a perfect
  matching $M$ on $T$ by matching $i$ and $j$ if and only if
  $w(i)=w(j)$ (cf.~Figure~\ref{fig:wheel}\,(a)). We define a subset
  $J\subset[n]$ with the intended meaning that if $i \in J$, we jump over
  the $i$-th step of $w$. The set $J$ will satisfy the following two
  conditions:
  \begin{itemize}
  \item every vertex $v$ with $F(v) = k$ is jumped over exactly once, i.e.\
    \[\bigl|\{i\in [n] \mid w(i)=v\}\cap J\,\bigr|=1,\quad \text{and}\]
  \item we never jump more than twice in a row, i.e.\ if $i,i+1\in J$, then
    $i+2\not\in J$.
  \end{itemize}

  \begin{figure}[tb]
    \begin{center}
      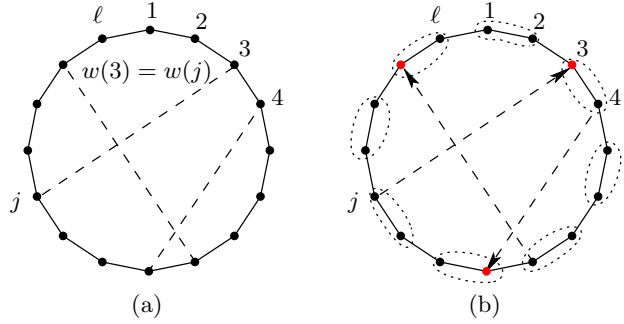
    \end{center}
    \caption{Deciding when to jump over vertices in a $k$-walk.}
    \label{fig:wheel}
  \end{figure}
  
  We partition the set $[n]$ into intervals of size $2$, setting
  \[U \coloneqq \big\{ \{1,2\},\{3,4\},\ldots \big\},\]
  with the last set $\{n\}$ being a singleton if $n$ is odd. Then the
  matching $M$ defines a multigraph without loops on $U$, and the degree of
  $I \in U$ is at most 2. We direct the edges of $M$, viewed as edges in
  the multigraph $(U,M)$, in such a way that every $I \in U$ has at most
  one incoming edge. The edges incident with $I$ correspond to the elements
  of $I \cap T$, and we put $i \in I$ into~$J$ if and only if the
  edge corresponding to $i$ is directed towards $I$
  (cf.~Figure~\ref{fig:wheel}\,(b)). For every
  $k=1,\ldots,\bigl\lfloor\frac12(n-1)\bigr\rfloor$ at most one of $2k-1$
  and $2k$ is in $J$, and therefore $J$ satisfies the above requirements.

  The definition of $P_k \subseteq V(G)$ is now straightforward:
  \[P_k\coloneqq
  \{ v\in V(G)\mid \text{$F(v)=k$ and $f(i)= k$ for the $i\in J$ with
    $w(i)=v$} \}.\]
  In summary, we end up with
  \[\sigma_k\coloneqq \{ E_{ab}\mid a,b\in [k] \}\cup
  \{ P_a\mid a=2,\ldots,k \},\]
  and it is clear that our construction can be carried out in polynomial
  time.
\end{proof}

\subsection{Proof of Theorem~\ref{thm:main}}

Let us finally derive the main theorem, Theorem~\ref{thm:main}. We first
need to draw upon the literature on model-checking first-order logic on
classes of bounded expansion. The following statement encapsulates the
model-checking results of Dvo\v{r}\'ak et al.~\cite{dvovrak2013testing} and
of Grohe and Kreutzer~\cite{grohe2011methods}. 
We also refer to the new expositions given in~\cite{Gajarski18, Pilipczuk18}. 

\begin{theorem}\label{thm:plain-FO-MC}
  Let $\tau$ be a finite and purely relational signature. Then for every
  formula $\phi\in \FO[\tau]$ there exists a nonnegative integer $r(\phi)$,
  computable from $\phi$, such that the following holds. Given a
  $\tau$-structure $\strA$, it can be verified whether $\strA\models \phi$
  in time $f\bigl(|\phi|,\adm_{r(\phi)}(G(\strA))\bigr)\cdot n$,
  where $n$ is the size of the universe of $\strA$ and $f$ is a computable
  function.
\end{theorem}

Observe that if $\strA$ is drawn from a fixed class of bounded expansion
$\CCC$, then $\adm_{r(\phi)}(G(\strA))$ is a parameter depending
only on $\phi$, hence we recover fixed-parameter tractability of
model-checking for \FO on any class of bounded expansion, parameterised by
the length of the formula. Theorem~\ref{thm:plain-FO-MC} is stronger than
this latter statement in that it says that the input structure does not
need to be drawn from a fixed class of bounded expansion, where the
colouring number is bounded in terms of the radius~$r$ for all values of
$r$, but it suffices to have a bound on the colouring numbers up to some
radius $r(\phi)$, which depends only on the formula $\phi$. We need
  this strengthening in our algorithm for the following reason. When adding
  a low-degree spanning tree to the Gaifman graph, we are not able to
  control all the colouring numbers at once, but only for some particular
  value of the radius. Theorem~\ref{thm:plain-FO-MC} ensures that this is
  sufficient for the model-checking problem to remain tractable.

We now sketch how Theorem~\ref{thm:plain-FO-MC} may be derived from the
works of Dvo\v{r}\'ak et al.\ \cite{dvovrak2013testing} and of Grohe and
Kreutzer~\cite{grohe2011methods}. We prefer to work with the algorithm of
Grohe and Kreutzer~\cite{grohe2011methods}, because we find it conceptually
simpler. For a given quantifier rank $q$ and an nonnegative integer
$i\le q$, the algorithm computes the set of all \emph{types}
$\mathfrak{R}_i^q$ realised by $i$-tuples in the input structure~$\strA$:
for a given $i$-tuple of elements $\overline{a}$, its \emph{type} is the
set of all \FO formulas $\phi(\overline{x})$ with~$i$ free variables and
quantifier rank at most $q-i$ for which $\phi(\overline{a})$ holds. Note
that for $i=0$ this corresponds to the set of sentences of quantifier rank
at most $q$ that hold in the structure, from which the answer to the
model-checking problem can be directly read; whereas for $i=q$ we consider
quantifier-free formulas with $q$ free variables. Essentially,
$\mathfrak{R}_q^q$ is computed explicitly, and then one inductively
computes $\mathfrak{R}_i^q$ based on $\mathfrak{R}_{i+1}^q$. The above
description is, however, a bit too simplified, as each step of the
inductive computation introduces new relations to the structure, but does
not change its Gaifman graph. We will explain this later.

When implementing the above strategy, the assumption that the structure is
drawn from a class of bounded expansion is used via \emph{treedepth-$p$
  colourings}, a colouring notion functionally equivalent to the
generalised colouring numbers. More precisely, a \emph{treedepth-$p$
  colouring} of a graph $G$ is a colouring $\gamma:V(G)\to\Gamma$, where
$\Gamma$ is a set of colours, such that for any subset $C\subseteq\Gamma$
of~$i$ colours, $i\le p$, the vertices with colours from~$C$ induce a
subgraph of treedepth at most~$i$. The \emph{treedepth-$p$ chromatic
  number} of a graph $G$, denoted $\chi_p(G)$, is the smallest number of
colours $|\Gamma|$ needed for a treedepth-$p$ colouring of~$G$. As proved
by Zhu \cite{zhu2009coloring}, the treedepth-$p$ chromatic numbers are
bounded in terms of \mbox{$r$-colouring} numbers as follows.

\begin{theorem}[\rm Zhu \cite{zhu2009coloring}]\label{lem:low-td-col}
  For any graph $G$ and $p\in\N$ we have
  \[\chi_p(G)\le \bigl(\col_{2^{p-2}}(G)\bigr)^{2^{p-2}}.\]
\end{theorem}

\medskip
Moreover, an appropriate treedepth-$p$ colouring can be constructed in
polynomial time from an ordering $L\in\Pi(G)$, certifying an upper bound on
$\col_{2^{p-2}}(G)$.

The computation of both $\mathfrak{R}_q^q$ and $\mathfrak{R}_i^q$ from
$\mathfrak{R}_{i+1}^q$ is done by rewriting every possible type as a purely
existential formula. Each rewriting step, however, enriches the signature
by unary relations corresponding to colours of some \mbox{treedepth-$p$}
colouring $\gamma$, as well as binary relations representing edges of
appropriate treedepth decompositions certifying that $\gamma$ is correct.
However, the binary relations are added in a way that the Gaifman graph of
the structure remains intact. For us it is important that in all the steps,
the parameter $p$ used for the definition of $\gamma$ depends only on $q$
and $i$ in a computable manner. Thus, by
Theorem~\ref{lem:low-td-col}, to ensure that $\gamma$ uses a bounded
number of colours, we only need to ensure the boundedness of
$\col_{r(q)}(G(\strA))$ for some computable function $r(q)$. By
taking~$q$ to be the quantifier rank of the input formula, the statement of
Theorem~\ref{thm:plain-FO-MC} follows.

\medskip
We can now combine all the ingredients and show how our main result follows
from Theorem~\ref{thm:tree-to-succ}.

\begin{proof}[Proof of Theorem~\ref{thm:main}] \ Given a
  successor-invariant formula $\phi\in \FO[\tau\cup\{S\}]$, we first
  compute the integer $r\coloneqq r(\phi)$ whose existence and
  computability is stated in Theorem~\ref{thm:plain-FO-MC}. Given the
  structure $\strA$, we now use the algorithm of
  Theorem~\ref{thm:adm-compute} to compute an order $L$ of the vertex
    set~$V(\strA)$ of the Gaifman graph of $\strA$, which satisfies
  $\adm_{80r+2}(G(\strA))\leq c(r)$ for some constant $c(r)$. Such
  constant exists by the assumption that $\strA$ is from a class of bounded
  expansion. We use the algorithm of Theorem~\ref{thm:main-technical} to
  compute a set of unordered pairs $F\subseteq\binom{V(\strA)}{2}$
  such that the graph $T=(V(\strA),F)$ is a tree of maximum degree
  at most $3$ and
  \[\adm_{40r+1}(G(\strA)+F,L)\le
  2+2\cdot\col_{80r+2}(G(\strA),L).\]
  By Lemma~\ref{lem:gen-col-ineq} we find
  $\col_{80r+2}(G(\strA),L)\leq
  \adm_{80r+2}(G(\strA),L)^{80r+2}\leq c(r)^{80r+2}$. This
    means that $\adm_{40r+1}(G(\strA)+F,L)\leq g(r)$ for
  $g(r)=2+2c(r)^{80r+2}$. Now, using the algorithm of
  Theorem~\ref{thm:tree-to-succ} we compute a successor relation $S$ such
  that
  \[\adm_r(G(\strA)+S,L)\le
  h\bigl(r,\adm_{40r+1}(G(\strA)+F)\bigr),\]
  where $h$ is the function from Theorem~\ref{thm:tree-to-succ}.
  Finally, we apply the algorithm of Theorem~\ref{thm:plain-FO-MC}
  to decide whether $(\strA,S)\models\phi$ in time
  $f\bigl(|\phi|,\adm_r(G(\strA)+S)\bigr)\cdot n$. Since $\strA$ is
  drawn from a fixed class of bounded expansion $\CCC$,
  $\adm_r(G(\strA)+S)$ is a parameter depending only on $\phi$. This
  finishes the proof of the theorem.
\end{proof}


\section{Dense Graphs}
\label{sec:dense}

While model-checking for first-order logic has been studied rather
thoroughly for sparse graph classes, few results are known for dense
graphs.
\begin{itemize}
\item On classes of graphs with bounded clique-width (or, equivalently,
  bounded rank-width; cf.~\cite{OumS06}), model-checking even for monadic
  second-order logic has been shown to be fixed-parameter tractable by Courcelle et
  al.~\cite{CourcelleMakRot00}.
\item More recently, model-checking on coloured posets of bounded width has
  been shown to be in fixed-parameter tractable for existential \FO by Bova et
  al.~\cite{bova2015model} and for all of \FO by Gajarsk\'y et
  al.~\cite{GajarskyHLOORS15}.
\end{itemize}

Both of these results extend to order-invariant \FO, and therefore also to
successor-invariant~\FO. For bounded clique-width, this has already been
shown in Section~\ref{sec:order-inv} . For posets
of bounded width we give a proof here. We first review the necessary
definitions.

\begin{definition}
  A \emph{partially ordered set (poset)} $(P,\leq^P)$ is a set $P$ with a
  reflexive, transitive and antisymmetric binary relation $\leq^P$. A
  \emph{chain} $C \subseteq P$ is a totally ordered subset, i.e.\ for all
  $x, y \in C$ one of $x \leq^P y$ and $y \leq^P x$ holds. An
  \emph{antichain} is a set $A \subseteq P$ such that if $x \leq^P y$ for
  $x,y \in A$, then $x = y$. The \emph{width} of $(P,\leq^P)$ is the
  maximal size $|A|$ of an antichain $A \subseteq P$.

  A \emph{coloured} poset is a poset $(P,\leq^P)$ together with a function
  $\lambda : P \to \Lambda$ mapping $P$ to some set $\Lambda$ of
  \emph{colours}.

  By $|P|$ we denote the length of a suitable encoding of $(P,\leq^P)$.
\end{definition}

We will need Dilworth's Theorem, which relates the width of a poset to the
minimum number of chains needed to cover the poset.

\begin{theorem}[Dilworth's Theorem]
  Let $(P,\leq^P)$ be a poset. Then the width of $(P,\leq^P)$ is equal to
  the minimum number $k$ of disjoint chains $C_i,\ldots,C_k\subseteq P$
  needed to cover $P$, i.e.\ such that $\bigcup_i C_i = P$.
\end{theorem}

A proof can be found e.g.\ in~\cite[Sec.~2.5]{diestel2012graph}. Moreover,
by a result of Felsner et al.~\cite{felsner2003recognition}, both the width
$w$ and a set of chains $C_1,\ldots,C_w$ covering $P$ can be computed from
$(P,\leq_P)$ in time $O(w\cdot|P|)$.

With this, we are ready to prove the following.

\begin{theorem}
  There is an algorithm which, on input a coloured poset $(P,\leq^P)$ with
  colouring $\lambda:P\to\Lambda$ and an order-invariant first-order
  formula $\varphi$, checks whether $P\models\varphi$ in time
  $f(w,|\varphi|)\cdot|P|^2$ where $w$ is the width of $(P,\leq^P)$.
\end{theorem}

\begin{proof}
  Using the algorithm of~\cite{felsner2003recognition}, we compute a chain
  cover $C_1,\ldots,C_w$ of $(P,\leq^P)$. To obtain a linear order on $P$,
  we just need to arrange the chains in a suitable order, which can be done
  by colouring the vertices with colours $\Lambda \times [w]$ via
  \[\lambda'(v)=(\lambda(v),j),\quad \text{for $v\in C_j$}.\]
  Then
  \begin{align*}
    \varphi_\leq(x,y)\coloneqq{}&
    \Bigl(
    \bigvee_{\substack{\lambda_x,\lambda_y \in \Lambda,\\i<j}}
    \bigl(\lambda'(x)=(\lambda_x,i)\wedge \lambda'(y)=(\lambda_y,j)\bigr)
    \Bigr)\\
    &\vee
    \Bigl(
    \bigvee_{\substack{\lambda_x,\lambda_y \in \Lambda,\\i\in [w]}}
    \bigl(\lambda'(x)=(\lambda_x,i)\wedge \lambda'(y)=(\lambda_y,i)\wedge
    x \leq y\bigr)\Bigr)
  \end{align*}
  defines a linear order on $(P,\leq_P)$ with colouring $\lambda'$. After
  substituting $\varphi_\leq$ for $\leq$ in $\varphi$, we may apply
  the algorithm of Gajarsk\'y et al.~\cite{GajarskyHLOORS15} to check whether
  $P\models \varphi$.
\end{proof}

\bibliographystyle{plain}
\bibliography{ref} 

\end{document}